\documentclass[runningheads]{llncs}

\usepackage{hyperref}
\usepackage{multicol}

\usepackage{tikz}

\title{Better Extension Variables in DQBF via Independence} 

\author{Leroy Chew \and Tom\'a\v{s} Peitl}
\institute{TU Wien, Austria}
\authorrunning{Leroy Chew, Tom\'a\v{s} Peitl}
\RequirePackage{amssymb,amsmath}
\RequirePackage{bussproofs/bussproofs}
\RequirePackage{xspace,xcolor}
\usepackage{url}
\usepackage{algorithm}
\usepackage{algorithmic}

\usepackage{xifthen}

\newlength{\breite}
\newcommand{\ComplexityClassFont}[1]{\mathsf{#1}}
\newcommand{\Pe}{\ComplexityClassFont{P}}
\newcommand{\PeP}{\ComplexityClassFont{P/poly}}

\newcommand{\PSPACE}{\ComplexityClassFont{PSPACE}}

\newcommand{\Math}[1]{\ensuremath{#1}\xspace}
\newcommand{\ProofSystemsFont}[1]{\mathsf{#1}}
\newcommand{\DefineProofSystem}[2]{\expandafter\def\csname#1\endcsname{\Math{\ComplexityClassFont{#2}}}}
\newcommand{\Red}{\Math{\boldsymbol{\forall}\kern .04ex\ProofSystemsFont{red}}}
\newcommand{\red}{\Math{\,\raisebox{.2ex}{$\scriptstyle+$}\,\Red}}
\newcommand{\Frege}[1][]{\Math{\ifthenelse{\isempty{#1}}{\ProofSystemsFont{Frege}}{#1\text{-}\ProofSystemsFont{Frege}}}}
\newcommand{\Res}[1][]{\Math{\ifthenelse{\isempty{#1}}{\ProofSystemsFont{Res}}{#1\text{-}\ProofSystemsFont{Res}}}}
\DefineProofSystem{eFrege}{e\Frege[]}
\DefineProofSystem{dFrege}{IndExt\Frege[]}
\DefineProofSystem{dRes}{IndExt QU\Res[]}
\DefineProofSystem{Gfull}{G}
\newcommand{\FregeRed}[1][]{\Frege\!\!\red}
\newcommand{\eFregeRed}[1][]{\eFrege\!\!\red}
\newcommand{\dFregeRed}[1][]{\dFrege\!\!\red}
\newcommand{\qrc}{\textsf{Q-Res}\xspace}

\newcommand{\irc}{\textsf{IR-calc}\xspace}

\newcommand{\dirc}{\textsf{DQBF-IR-calc}\xspace}
\newcommand{\idrc}{\textsf{IR(\Drrs)-calc}\xspace}

\newcommand{\Dz}[1]{\ensuremath{\mathcal{D}^{\mbox{\scriptsize \upshape #1}}}\xspace}

\newcommand{\Drrs}{\Dz{rrs}}
\newcommand{\Dstd}{\Dz{std}}
\newcommand{\lqrc}{\textsf{LD-Q-Res}\xspace}
\newcommand{\qdrc}{\textsf{Q(\Drrs)-Res}\xspace}
\newcommand{\qsrc}{\textsf{Q(\Dstd)-Res}\xspace}
\newcommand{\lqdrc}{\textsf{LD-Q(\Drrs)-Res}\xspace}

\newcommand{\qurc}{\textsf{QU-Res}\xspace}
\newcommand{\ecalculus}{$\forall$\textsf{Exp+Res}\xspace}
\newcommand{\irmc}{\textsf{IRM-calc}\xspace}

\newcommand{\lqurc}{\textsf{LQU}\textsf{-\hspace{1pt}Res}\xspace}
\newcommand{\lquprc}{\textsf{LQU}$^+$\textsf{-\hspace{1pt}Res}\xspace}

\newcommand{\qrat}{\textsf{QRAT}\xspace}

\newcommand{\instant}{\mbox{\textsf{\upshape inst}}}
\newcommand{\axiom}{\mbox{\textsf{\upshape axiom}}}

\newcommand{\res}{\mbox{\textsf{\upshape res}}}

\newcommand{\propvarset}{\mathcal{V}}
\DeclareMathOperator*{\var}{\operatorname{var}}
\newcommand{\CNF}{\textsc{CNF}}

\DeclareMathOperator*{\lv}{\operatorname{lv}}

\newcommand{\pathc}{\mathfrak{C}}
\newcommand{\pathl}{\mathfrak{L}}

\DeclareMathOperator*{\instantiate}{\textsf{inst}}

\DeclareMathOperator*{\duality}{\operatorname{ \mathsf{Duality}}}
\DeclareMathOperator*{\select}{\operatorname{\mathsf {Select}}}

\tikzstyle{uredge}=[very thick,draw=red!80!green,line cap=round]
\tikzstyle{redge}=[line cap=round,very thick,draw=red!30!green!20!blue
]

\tikzstyle{axiomn}=[rectangle,very thick,draw=black!50,%
top color=white,bottom color=green!50,
minimum height=1.1em,minimum width=.5cm,inner sep=2pt,draw%
]

\tikzstyle{calcn}=[rectangle%
,rounded corners=2mm%
,thick%
,draw=black%
,top color=white,bottom color=black!20,%
minimum height=.5cm,minimum width=.5cm,inner sep=2pt%
]

\usetikzlibrary{calc}

\tikzstyle{infn}=[rectangle,rounded corners=1mm,thick,draw=black!50,%
top color=white,bottom color=black!15,%
minimum height=1.3em,minimum width=.5cm,inner sep=2pt,draw%
]

\tikzstyle{botn}=[rectangle,rounded corners=1mm,very thick,draw=black,%
top color=white,bottom color=red!25,%
minimum height=.5cm,minimum width=.5cm,inner sep=2pt,draw%
]

\tikzstyle{expcalcn}=[rectangle%
,thick%
,draw=green!40!black
,top color=white,bottom color=green!20,%
minimum height=.5cm,minimum width=.5cm,inner sep=2pt%
]

\tikzstyle{algcalcn}=[rectangle%
,thick%
,draw=yellow!40!black
,top color=white,bottom color=yellow!20,%
minimum height=.5cm,minimum width=.5cm,inner sep=2pt%
]

\tikzstyle{strongcalcn}=[rectangle%
,thick%
,draw=blue!40!black
,top color=white,bottom color=blue!20,%
minimum height=.5cm,minimum width=.5cm,inner sep=2pt%
]

\tikzstyle{cdclcalcn}=[rectangle%
,thick%
,draw=magenta!40!black
,top color=white,bottom color=magenta!20,%
minimum height=.5cm,minimum width=.5cm,inner sep=2pt%
]

\tikzstyle{expcalcn}=[rectangle%
,thick%
,draw=green!40!black
,top color=white,bottom color=green!20,%
minimum height=.5cm,minimum width=.5cm,inner sep=2pt%
]

\tikzstyle{strcalcn}=[rectangle%
,thick%
,draw=pink!40!black
,top color=white,bottom color=pink!20,%
minimum height=.5cm,minimum width=.5cm,inner sep=2pt%
]

\newlength{\radius}
\tikzstyle{decision} = [diamond, draw, fill=blue!20, 
text width=4.5em, text badly centered, node distance=3cm, inner sep=0pt]
\tikzstyle{block} = [rectangle, draw, fill=blue!20, 
text width=5em, text centered, rounded corners, minimum height=4em]
\tikzstyle{line} = [draw, -latex']
\tikzstyle{cloud} = [draw, ellipse,fill=red!20, node distance=3cm,
minimum height=2em]

\definecolor{violet}{RGB}{138,43,226}
\definecolor{forestgreen}{RGB}{34,139,34}
\definecolor{darkblue}{RGB}{102,0,204}
\definecolor{pink}{RGB}{255,192,20}
\definecolor{gold}{RGB}{255,215,0}

\begin{document}
	\maketitle
	\begin{abstract}
		We show that extension variables in (D)QBF can be generalised by conditioning on universal assignments. The benefit of this is that the dependency sets of such conditioned extension variables can be made smaller to allow easier refutations. This simple modification instantly solves many challenges in p-simulating the QBF expansion rule, which cannot be p-simulated in proof systems that have strategy extraction \cite{ChewCly20}. 
		Simulating expansion is even more crucial in DQBF, where other methods are incomplete. In this paper we provide an overview of the strength of this new independent extension rule. We find that a new version of Extended Frege called \dFregeRed can p-simulate a multitude of difficult QBF and DQBF techniques, even techniques that are difficult to approach with \eFregeRed.
		We show six p-simulations, that \dFregeRed p-simulates QRAT, \dirc, \idrc, Fork-Resolution, DQRAT and \Gfull which together underpin most DQBF solving and preprocessing techniques. The p-simulations work despite these systems using complicated rules and our new extension rule being relatively simple. Moreover, unlike recent p-simulations by \eFregeRed we can simulate the proof rules line by line, which allows us to mix QBF rules more easily with other inference steps.
		
	\end{abstract}
	
\section{Introduction}

Proof systems that allow extension variables are very powerful~\cite{CR79}.
We know that in propositional logic, extended resolution can p-simulate a multitude of disparate techniques, and lower bounds to extended resolution remain an open problem.
This is remarkable because resolution itself is a weak proof system with various lower bounds, but simply adding the extension rule upgrades it to be amongst the most powerful propositional proof systems and is equivalent to the checking format DRAT \cite{WetzlerHH14}. 
An extension rule for variable $v$ takes the form:
$$ v \leftrightarrow b(X): b \text{ is a Boolean function, } X \text{ is a set of existing variables} $$ 

Quantified Boolean Formulas (QBF) also use extension variables \cite{Jus07} and these too can be very powerful. 
QBFs list all variables in a quantifier prefix. This defines whether a variable is existential ($\exists$) or universal ($\forall)$. The order defines dependencies so variables only depend on other variables to their left.
When used in refutations, extension variables must be existential, otherwise it is too easy to violate extension clauses. 
The dependency set of the new variables (which is the same as the quantification order in QBF) must be careful not to introduce falsity into the formula. In fact one of the earliest approaches was to say that every new variable is quantified rightmost to conservatively give it the entire dependency set. The drawback with this approach is that it limits the power of inferences from using extension variables \cite{BCJ16ext}. In the current alternative, sometimes known as ``strong extension'', we place $\exists v$ immediately after all the variables $X$ used in $b$.
\eFregeRed uses this strong notion, and it is so powerful no unconditional lower bound can be found, unless either a long standing proof complexity or circuit complexity open problem is solved \cite{BBCP20}. 
Some problems do emerge when looking at conditional lower bounds, though. The $\select$ family of formulas can be shown by the QBF proof system QRAT to be equivalent to the law of non-contradiction on QBFs, but do not have short proofs in QBF Extended Frege unless $\PSPACE\subseteq\PeP$ \cite{ChewCly20}.
QRAT manages this through a combination of extension variables and an explicit rule that calculates so-called spurious dependencies in quantified variables,  when nominal dependencies can be ignored when making inferences.
Note how different this is to propositional logic, in most cases where we would take a basic proof system and add even a complicated rule we would still be simulated by extended resolution.
Based on our observations in this paper we find that we can do better with extension variables. Instead of extensions being pure definitions, they are now under conditions. 
$$ \alpha \rightarrow (v \leftrightarrow b(X)): \alpha \text{ is a partial assignment of the universal variables} $$
The utility of this is not immediately obvious because $v$ is weaker than it could be, however precisely because this is weaker now the dependency of $v$ on the variables of $\alpha$ is no longer necessary for soundness. 
Because we can remove arbitrary dependencies the natural class for this proof system is for Dependency QBF (DQBF). But we can still use this rule for QBF, in fact it adds substantial clarity to existing QBF proof systems.

Our main contribution is that we propose a line based proof system, that each new line addition preserves satisfiability, so when arriving at the falsum symbol $0$ we know we started with an unsatisfiable DQBF. Unlike many QBF systems we are not able to construct countermodels from following the proof steps (unless $\Pe=\PSPACE$). We then show how it can p-simulate existing QBF and DQBF proof system rules.
Figure~\ref{fig:QBFsimstruct} shows the known p-simulations in QBF proof systems after considering the work in this paper.

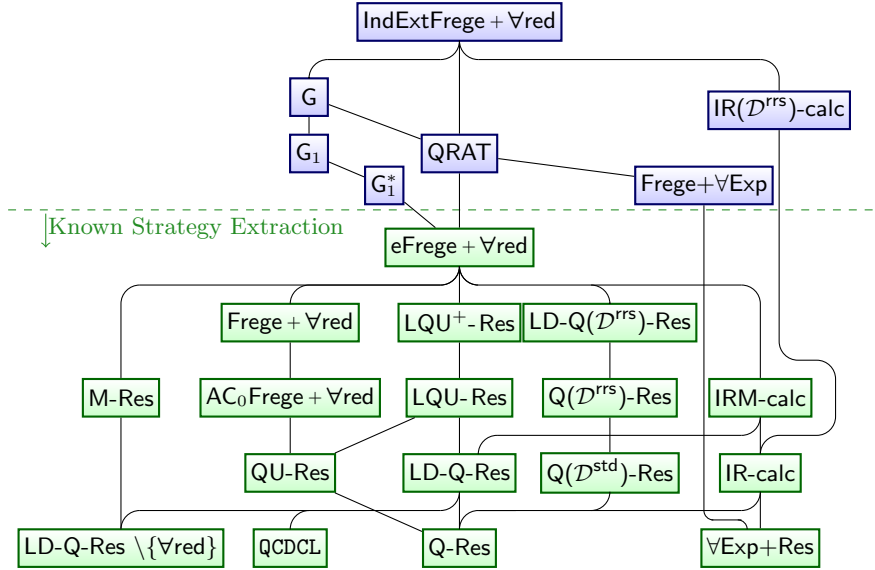
\begin{figure}
	\begin{tikzpicture}
	\node[strongcalcn](DFred) at (0,10){\dFregeRed};
	\node[strongcalcn](G) at (-2,9){\Gfull};
	\node[strongcalcn](G1) at (-2,8.25){\textsf{G}$_1$};
	\node[strongcalcn](G1t) at (-1,7.825){\textsf{G}$^*_1$};
	\node[strongcalcn](QRAT) at (0,8.25){\qrat};
	\node[expcalcn](EFred) at (0, 7){\eFregeRed};
	\node[strongcalcn](Fexp) at (3.25, 7.825){\textsf{Frege}+$\forall$\textsf{Exp}};
	\node[strongcalcn](idrc) at (4.25, 8.825){\idrc};
	\node[expcalcn](Fred) at (-2.25,6){\FregeRed};
	\node[expcalcn](Cred) at (-2.25,5){$\mathsf{AC}_0$\FregeRed};
	\node[expcalcn](LQM) at (-4.5,3){\lqrc$\setminus\{\Red\}$};
	\node[expcalcn](M) at (-4.5,5){\textsf{M-Res}};
	\node[expcalcn](QU) at (-2.25,4){\qurc};
	\node[expcalcn](Q) at (0,3){\qrc};
	\node[expcalcn](LQ) at (0, 4){\lqrc};
	\node[expcalcn](LQU) at (0, 5){\lqurc};
	\node[expcalcn](LQUP) at (0, 6){\lquprc};
	\node[expcalcn](QS) at (2, 4){\qsrc};
	\node[expcalcn](QD) at (2, 5){\qdrc};
	\node[expcalcn](LQD) at (2, 6){\lqdrc};
	\node[expcalcn](QCDCL) at (-2.25, 3){\texttt{QCDCL}};
	\node[expcalcn](e) at (4,3){\ecalculus};
	\node[expcalcn](ir) at (4,4){\irc};
	\node[expcalcn](irm) at (4,5){\irmc};
	
	\draw[dashed, forestgreen](-6,7.5)--(5.5,7.5);
	\draw[->,forestgreen](-5.5,7.45)--(-5.5,7);
	\node(selabel) at (-3.5,7.25){\textcolor{forestgreen}{Known Strategy Extraction}};
	\draw(Q)--(QU)--(Cred)--(Fred);
	\draw(LQUP)--(EFred)--(QRAT);
	\draw(e)--(ir)--(irm);
	\draw(DFred)--(QRAT)--(Fexp);
	\draw(idrc)--(4.25, 5.75) arc (0:90:-0.5\radius)--(4.75, 5.5) arc (270:180:-0.5\radius)--(5, 4.75) arc (0:-90:0.5\radius)--(4.25,4.5) arc (90:180: 0.5 \radius);
	\draw(Fexp)--(3.25, 3.5);
	
	\draw (3.25, 3.5) arc(0:90:-0.25\radius) (3.35,3.375) ;
	\draw (3.35,3.375)--(3.75,3.375);
	\draw (3.70,3.375) arc(-90:-180:-0.25\radius) (e) ;
	\draw(QD)--(LQD);
	\draw(QS)--(QD);
	\draw(1.75,3.5) arc(-90:0:0.5\radius) (QS);
	
	\draw(LQ)--(LQU)--(LQUP);
	
	
	\draw(0.25,3.5)--(3.75,3.5);
	\draw(3.75,3.5) arc(-90:0:0.5\radius) (ir);
	\draw(-0.25,3.5) arc(-90:0:0.5\radius) (ir);
	
	\draw(0.5,4.5)--(3.75,4.5);
	\draw(0.5,4.5) arc(90:180:0.5\radius) (LQ);
	\draw(3.75,4.5) arc(-90:0:0.5\radius) (irm);
	
	\draw(-2,6.5)--(-0.25,6.5);
	\draw(-2,6.5) arc(90:180:0.5\radius) (LQUP);
	\draw(-0.25,6.5) arc(-90:0:0.5\radius)-- (EFred);
	\draw(-2,3.5) arc(90:180:0.5\radius);
	
	\draw(Q)--(LQ);
	
	\draw(3.75,6.5)--(0.25,6.5);
	\draw(3.75,6.5) arc(90:0:0.5\radius) --(irm);
	\draw(1.75,6.5) arc(90:0:0.5\radius) --(LQD);
	\draw(0.25,6.5) arc(270:180:0.5\radius)-- (EFred);
	
	\draw(-4.25,6.5)--(-0.25,6.5);
	\draw(-4.25,6.5) arc(90:180:0.5\radius)--(M)--(LQM);
	\draw(-0.25,6.5) arc(-90:0:0.5\radius)-- (EFred);

	\draw(-0.25,3.5)--(-4.25,3.5);
	\draw(-4.5,3.25) arc(180:90:0.5\radius) (LQ);
	\draw(-1.75,9.5) arc(90:180:0.5\radius)--(G);
	\draw(-1.75,9.5)--(-0.25,9.5);
	\draw(-0.25,9.5) arc(-90:0:0.5\radius)--(DFred);
	\draw(0.25,9.5) arc(270:180:0.5\radius)-- (DFred);
	\draw(4, 9.5) arc(90:0:0.5\radius)--(idrc);
	\draw(0.25,9.5)--(4, 9.5);
	\draw(QU)--(LQU);
	
	\draw(G1)--(G);
	\draw(EFred)--(G1t)--(G1);
	\draw(QRAT)--(G);
	

	\draw(0.25,6.5) arc(270:180:0.5\radius)-- (EFred);
	
	\draw(0.25,3.5)--(0.75,3.5);
	\draw(0.25,3.5) arc(90:180:0.5\radius) (Q);
\end{tikzpicture}
\caption{The p-simulation structure of refutational QBF proof systems \label{fig:QBFsimstruct}}
\end{figure}

\subsection{Related Work}

The original Extended Q-Resolution was defined by Jussila et al. \cite{Jus07} and discussed the weak and strong extension variables.
QRAT by Biere, Heule and Seidl generalises extension variables in its clause additions. In the original paper \cite{HSB14} and later in Kiesl and Seidl's paper \cite{KS19}, they use similar clause additions to our independent extension definitions. Unlike our definitions, the new variables have a larger dependency set, but get round this by detecting spurious dependencies with an Extended Universal Reduction rule (EUR). 
These QRAT extension variables can capture the annotated
variables used in the QBF proof system $\forall\textsf{Exp+Res}$, but there has not been success
making the same method work for the more general QBF proof system \irc \cite{chede2021does}.
Our approach works backwards, instead of using the presence of resolution path schemes to derive expansions, we use expansions to show the validity of resolution path schemes. 

Blinkhorn proposed a DQBF proof system based on generalising the QBF proof system QRAT to DQBF \cite{Blinkhorn2020SimulatingDP}. This system too allows the addition of new variables $v$, but  when using the RAT addition rules to make definition clauses the RAT addition rule only allows the variables in the dependency set of $v$ to be used unlike in our case. Therefore we suspect their proof system is weaker. Rabe proposed DQBF proof system called Fork-Resolution that proposed extension variables but only in the case of clause-splitting \cite{rabe17}.

Chew and Clymo showed the QRAT's strongest rule violated the property of strategy extraction unless $\Pe=\PSPACE$ \cite{ChewCly20}. Later, Chew and Heule showed that the QBF sequent calculus \Gfull p-simulates QRAT \cite{CH22}. The sequent calculus \Gfull\cite{KP90} creates quantified variables of propositional and even QBF witnesses so it 
has witnessing scheme that works to simulate EUR. We use the same witnessing scheme in our work to simulate EUR with our DQBF system, hence showing strategy extraction, even for the QBF fragment, is impossible unless $\Pe=\PSPACE$.

The discussion around new variables comes after Reichl and Slivovsky's successful DQBF solver \texttt{Pedant} \cite{reichl2022pedant} introduces the similar \textit{arbiter variables}, with a few more limitations.

\subsection{Organisation}\label{sec:org}

We define the necessary preliminaries of QBF and DQBF in Section~\ref{sec:prelim}. As an example we include the DQBF complete refutation system \dirc. 
In Section~\ref{sec:maindef} we define our new extension rule and show its soundness. We integrate it into a full proof system that we call \dFregeRed. In order to show completeness we p-simulate the complete \dirc proof system.
Section~\ref{sec:psim} demonstrates the advantages of using \dFregeRed on (D)QBFs, with more p-simulations. 
In Section~\ref{sec:dres} we demonstrate a p-equivalent definition which we name \dRes based on resolution.
In Section~\ref{sec:fork} we give a definition of the clausal proof system Fork-Resolution for DQBF and show how to p-simulate that in \dFregeRed.
In Section~\ref{sec:qrat} we give a definition of the QBF interference based proof system QRAT and show how \dFregeRed can p-simulate it. In Section~\ref{sec:qdrc} we do the same for the QBF proof system \idrc which uses a reflexive resolution path dependency scheme.
Finally in Section~\ref{sec:G} we show a p-simulation of the sequent calculus \Gfull.

\section{Preliminaries}\label{sec:prelim}




We assume a countably infinite set $\propvarset$ of propositional \emph{variables} is given.
A \emph{literal} is either a variable $x \in \propvarset$ or its negation ($\neg x$), also written as $\bar x$, whereby $\bar {\bar x} = x$.
A \emph{(propositional) formula} is defined recursively: (1) literals are formulas; and (2) if $\phi$ and $\psi$ are formulas, then $\phi \land \psi$, $\phi \lor \psi$, $\phi \rightarrow \psi$, $\phi \leftrightarrow \psi$ and $\bar \phi$ are also formulas.
A \emph{circuit} is like a formula, but its recursive structure is a DAG (directed acyclic graph) instead of a tree: subformulas can be reused.
Every formula is a circuit, but transforming a circuit into a formula involves potentially exponential duplication of nodes.
The set of variables of a circuit or formula is defined recursively as $\var(\phi \circ \psi) = \var(\phi) \cup \var(\psi)$ for any operator $\circ$, $\var(\bar \phi) = \var(\phi)$, and $\var(x) = x$ if $x$ is a variable.
We use $\sqcup$ to denote the disjoint union, i.e. the union of two sets that are known to be disjoint.
A \emph{clause} is a finite set of literals, semantically interpreted as their disjunction (equivalently, a formula consisting only of literals and ``or'' connectives).
A clause $C$ is \emph{tautological}, if $\{x, \bar x\} \in C$ for some variable $x$. A clause is \emph{unit} if only contains a single literal, a clause is the \emph{empty clause} if it contains no literals. The negation $\bar C$ of a clause $C$ can be viewed as a conjunction of unit clauses $\bigwedge_{l\in C} \{\bar l\}$.

A \emph{(partial) assignment} to a circuit $\phi$ is a mapping $\alpha : V \subseteq \var(\phi) \to \{0, 1\}$.
A partial assignment $\alpha$ is \emph{complete} if it assigns every variable, i.e.\ $V = \var(\phi)$.
We can write an assignment $\alpha$ as a function, a conjunction of literals or sequence of literals i.e. $x\bar y\bar z$. 
As a function an assignment is extended to circuits by $\alpha(\phi \land \psi) = \alpha(\phi) \cdot \alpha(\psi)$, $\alpha(\phi \lor \psi) = \max ( \alpha(\phi), \alpha(\psi) )$ and $\alpha(\bar \phi) = 1 - \alpha(\phi)$.
$\alpha$ \emph{satisfies} (is a \emph{model} of) a circuit $\phi$ if $\alpha(\phi) = 1$.
A circuit is \emph{satisfiable} if it has a model, and \emph{unsatisfiable} otherwise.
Circuit $\phi$ is a \emph{tautology} if $\bar \phi$ is unsatisfiable.
Circuit $\phi$ \emph{entails} a circuit $\psi$, written $\phi \models \psi$, if every complete assignment $\alpha : \var(\phi) \to \{0,1\}$ that satisfies $\phi$ also satisfies $\psi$.
Two circuits $\phi$ and $\psi$ are \emph{(logically) equivalent}, written $\phi \equiv \psi$, if $\var(\phi) = \var(\psi)$, and $\phi \models \psi$ and $\psi \models \phi$.
Two circuits $\phi$ and $\psi$ are \emph{equisatisfiable}, written $\phi \cong \psi$, if they are both satisfiable or both unsatisfiable.

A formula is in \emph{conjunctive normal form (CNF)} if it is a (finite) conjunction of clauses.
Any circuit $\phi$ can be transformed into a logically equivalent CNF using distributivity and De Morgan's rules, but the resulting size may be exponential.
Allowing \emph{extension variables}, i.e.\ variables that do not occur in $\phi$, it is possible to transform any circuit $\phi$ into $\CNF(\phi)$, an equisatisfiable CNF in linear time~\cite{tseitin68}.
This translation has the additional property that if $\alpha$ is a model of $\CNF(\phi)$, then $\alpha|_{\var(\phi)}$ is a model of $\phi$ and vice versa, if $\alpha$ is a model of $\phi$, then there exists a model $\beta$ of $\CNF(\phi)$ such that $\beta(x) = \alpha(x)$ for all $x \in \var(\phi)$.

\subsection{Proof systems}

A \emph{proof system} as defined by Cook and Reckhow~\cite{CR79} for some non-empty language $\mathcal{L}$ is a polynomial-time computable function on strings whose range is exactly $\mathcal{L}$. Intuitively $f$ maps proofs to valid theorems, non-proofs can be mapped to some arbitrary known element in $\mathcal{L}$. Soundness comes from well-definition, and completeness is from surjectivity. Many proof systems are line-based where a finite set of rules govern the derivation of valid inferences in the language until a conclusive line is derived. 

Proof complexity measures the \emph{sizes} of proofs, i.e. number of characters in the proof string. In a line based proof systems, where we can distinguish individual lines we can measure the proof \emph{length}- the number of lines. Where lines are clauses, the clause \emph{width} is the number of literals in the clause. Given a proof system $g$ for language $\mathcal{L}$ we say that proof system $f$ \emph{p-simulates} $g$ if there is a polynomial time procedure $p$ that maps $g$ proofs to $f$ proof such that $f(p(\pi))= g(\pi)$ for all $g$-proofs. We do not require $f$ to necessarily be a proof system for $\mathcal{L}$, but a proof system for $\mathcal{L}'\supseteq\mathcal{L}$. In this paper we do this when we p-simulate QBF proof systems with DQBF proof systems. 

Frege systems are ``text-book'' style line-based proof systems for propositional logic. They consist of a finite, sound and complete set of axioms and rules where any variable can be substituted by any formula (such as Modus Ponens). 
We give an example of a Frege system in Figure~\ref{fig:frege1}.
\begin{figure}[h]
	\framebox{\parbox{0.98\textwidth}
		{
			\begin{prooftree}
				\AxiomC{}
				\UnaryInfC{$1$}
				\DisplayProof\hspace{2cm}
				\AxiomC{}
				\UnaryInfC{$x_1\rightarrow (x_2 \rightarrow x_1)$}
				\DisplayProof\hspace{2cm}
				\AxiomC{}
				\UnaryInfC{$((x_1\rightarrow 0)\rightarrow 0)\rightarrow x_1$}
			\end{prooftree}
			\begin{prooftree}
				
				\AxiomC{}
				\UnaryInfC{$(x_1\rightarrow (x_2 \rightarrow x_3))\rightarrow ((x_1\rightarrow x_2)\rightarrow(x_1\rightarrow x_3))$}
				\DisplayProof\hspace{1cm}
				\AxiomC{$x_1$}
				\AxiomC{$x_1\rightarrow x_2$}
				\BinaryInfC{$x_2$}
			\end{prooftree}
			\begin{prooftree}
				\AxiomC{}
				\UnaryInfC{$(x_1\rightarrow 0)\rightarrow \neg x_1$}
				\DisplayProof\hspace{1cm}
				\AxiomC{}
				\UnaryInfC{$ \neg x_1 \rightarrow (x_1\rightarrow 0)$}
				\DisplayProof\hspace{1cm}
				\AxiomC{}
				\UnaryInfC{$ (x_1 \vee x_2) \rightarrow (x_2\vee x_1) $}
			\end{prooftree}
			\begin{prooftree}
				\AxiomC{}
				\UnaryInfC{$x_1\rightarrow (x_1\vee x_2)$}
				\DisplayProof\hspace{1cm}
				\AxiomC{}
				\UnaryInfC{$(x_1 \vee x_2)\rightarrow (\neg x_1 \rightarrow x_2)$}
				\end{prooftree}
				\begin{prooftree}
				\AxiomC{}
				\UnaryInfC{$(x_1 \rightarrow x_2)\rightarrow (\neg x_1 \vee x_2)$}
				\DisplayProof\hspace{1cm}
				\AxiomC{}
				\UnaryInfC{$(x_1\wedge x_2)\rightarrow x_1$}
				\end{prooftree}
				\begin{prooftree}
				\AxiomC{}
				\UnaryInfC{$(x_1\wedge x_2)\rightarrow x_2$}
				\DisplayProof\hspace{1cm}
				\AxiomC{}
				\UnaryInfC{$ x_1 \rightarrow (x_2\rightarrow x_1 \wedge x_2)$}
			\end{prooftree}
	}}
	\caption{A \Frege system for connectives $0,1, \neg, \rightarrow \vee, \wedge $ \label{fig:frege1}.}
	
\end{figure}

The rules will depend on the connectives included, but Cook and Reckhow \cite{CR79} showed all Frege systems are p-equivalent. 
For example a Frege system that uses $\vee, \neg, \wedge$ can adopt the following laws without changing the proof complexity:
\begin{prooftree}
	\AxiomC{$C$}
	\RightLabel{(Weak)}
	\UnaryInfC{$C \vee p$}
	\DisplayProof\hspace{1cm}
	\AxiomC{$C\vee p$}
	\AxiomC{$C\vee \neg p$}
	\RightLabel{(Res)}
	\BinaryInfC{$C$}
	\DisplayProof\hspace{1cm}
	\AxiomC{$C\vee p$}
	\AxiomC{$C\vee q$}
	\RightLabel{(Distr)}
	\BinaryInfC{$C \vee (p\wedge q)$}
\end{prooftree}

\subsection{(Dependency) Quantified Boolean Formulas}
A \emph{quantified Boolean formula (QBF)} is a propositional formula equipped with Boolean quantifiers: $\forall $ and $\exists$. $\forall x \phi(x)\equiv \phi(0) \wedge \phi(1)$ and $\exists x \phi(x)\equiv \phi(0) \vee \phi(1)$. A QBF in prenex form $\Pi\phi$ contains a propositional \emph{matrix} $\phi$ which is quantifier-free and a \emph{prefix} $\Pi= \mathcal{Q}_1 x_1 \dots \mathcal{Q}_k x_k$ where $\mathcal{Q}_i\in \{\forall, \exists\}$ for $1 \leq i \leq k$. 
A closed QBF requires every variable to be bound to some quantifier in the prefix. We will mainly work with closed prenex QBFs in this paper. 
For variable $x$ we use $x\in \exists$ to denote that $x$ is existentially quantified somewhere in the prefix, or $x\in \forall$ to denote that $x$ is universally quantified in the prefix. The quantifier order matters, we say that existential variable $x$ depends on $u$ if $u$ in quantified left on $x$ in the prefix. In this way we can build a dependency set $D_x$ of each existentially quantified variable $x$ containing exactly the universal variables that are left of $x$.

A \emph{Skolem function} for existential variable $x$ is a Boolean function $f_x:D_x \rightarrow \{0,1\}$. This allows us to use alternative semantics to define the truth of a QBF: that a closed prefix QBF is true if and only if there is a set of Skolem functions, one for each existential variable $x$, such that for every complete assignment to all the universal variables the universal assignment completed with the values of the Skolem functions under that assignment, form a satisfying assignment to the propositional matrix. We call such a set of Skolem functions \emph{winning}. 
Dually, a closed QBF  is false if and only if there is a set of \emph{Herbrand functions}, one for each universal variable $y$, such that for every complete assignment to all the existential variables the Herbrand functions falsify the propositional matrix.

A Dependency Quantified Boolean Formula (QBF) can be defined and an S-form DQBF uses this notion of Skolem functions as its main semantics. An S-form DQBF $\Pi\phi$ has a prefix $\Pi= \forall u_1 \dots u_p \exists {x_1}(D_{x_1}) \dots {x_q}(D_{x_q})$, here the quantifier order does not matter as the dependency sets are explicitly given. Each $D_x$ can be any arbitrary subset of $\{u_1 \dots u_p\}$.
A DQBF is true if and only if there is a set of Skolem functions, one for each existential variable $x$, such that for every complete assignment to all the universal variables the universal assignment completed with the values of the Skolem functions under that assignment, form a satisfying assignment to the proposition matrix.
We sometimes write $\forall U \exists E \phi$ for an arbitrary S-form DQBF, where $U$ is the set of universal variables, $E$ the set of existential variables each with their own unspecified dependency set and $\phi$ a propositional matrix containing no quantifiers.
We define $D_u$ for a universal variable to be $\{u\}$. We can also define the dependency set of  a clause $C$ as  $(\bigcup_{\var(y)\in C} D_y)$ including universal literals.


QBF is a \textsf{PSPACE}-complete language and DQBF is \textsf{NEXPTIME}-complete. We can demonstrate a DQBF is true by exhibiting its Skolem functions as circuits and showing the matrix with the Skolem functions substituted in is a propositional tautology. To show a DQBF is false we can use a DQBF proof system such as \dirc \cite{BeyersdorffBlinkhornChewSchmidtSuda19} given in Figure~\ref{fig:irc}. We will define a new refutational proof system for DQBF in the next section.

\section{{\dFregeRed}, an S-form DQBF proof system}\label{sec:maindef}

\subsection{Independent Extension}

Consider the refutational proof system \FregeRed in QBF \cite{BBCP20}, composed of Frege rules and a $\forall$\textsf{red} rule. Frege rules allow for propositional line based inference and the reduction rule allows a universal variable $u$ to be replaced by a constant $0$ or $1$, given below:

\begin{prooftree}
	\AxiomC{$\Pi\phi \wedge L(u)$}
	\UnaryInfC{$\Pi\phi \wedge L(u)\wedge L(0)$}
	\DisplayProof
	\quad
	\AxiomC{$\Pi\phi \wedge L(u)$}
	\UnaryInfC{$\Pi\phi \wedge L(u)\wedge L(1)$}
\end{prooftree}

The side condition is that no variable appears to the right of $u$ in $L$. 
A sound DQBF version exists (proof in the Appendix).

\begin{lemma}[$\forall$\textsf{red} soundness]\label{lem:redsound}
	Suppose $\Pi \phi \wedge L(u)$ is a true S-form DQBF, and $L$ contains no existential variables $x$ such that $u$ is in the dependency set of $x$.
	Then the S-form DQBF $\Pi \phi \wedge L(u)\wedge L(0)$ is true and the S-form DQBF $\Pi \phi \wedge L(u)\wedge L(1)$ is also true.
\end{lemma}
	\begin{proof}
	An S-form DQBF is true if and only if it has a set of satisfying Skolem functions, a function $\sigma_x$ for each of its existential variables $x$.
	Suppose $\Pi \phi \wedge L(u)$ is satisfied by the set $\{\sigma_x \mid x\in \exists\}$.
	We will show that $\Pi \phi \wedge L(u)\wedge L(0)$ is satisfied by the same functions. 
	Consider an arbitrary universal assignment $\alpha$, $\phi$ is satisfied by assumption and so is $L(u)$. If $\alpha(u)$ is $0$ then $L(u)=L(0)$ is satisfied.
	Otherwise consider $\beta$ which is identical to $\alpha$ except on $u$. $\alpha(u)=1$ and $\beta(u)=0$, but the outputs of Skolem functions of the existential variables in $L$ remain unaffected by changing between $\alpha$ and $\beta$, only the variable $u$ is affected and so $L(0)$ is satisfied.
	Therefore $\Pi \phi \wedge L(u)\wedge L(0)$ is satisfied by the same set of  Skolem functions as $\Pi \phi \wedge L(u)$. The case with $L(1)$ is symmetrical.\qed
\end{proof}
%

Universal reduction in $L(u)$ is blocked when there are existential variables in $L$ that depend on $u$. Extension variables are also existential and can end up blocking reduction through excessive dependency. We define extension variables that \textit{conditionally} represent Boolean circuits, for smaller dependency sets. We give two versions, one for conjunction and one for disjunction. We could instead use a single rule based on a functionally complete connective such as NAND 
(which we do in Section~\ref{sec:dres}), 
 but our definitions fit more nicely into the proofs of this paper.
Let $\alpha$ be a conjunction of universal literals, and $Y$ is a set of literals, both existential and universal. 

\begin{prooftree}
	\AxiomC{$\Pi\phi$}
	\UnaryInfC{$\Pi \exists v(D_v) \phi \wedge (\alpha \rightarrow (v \leftrightarrow \bigwedge_{y\in Y}  y))$}
	\DisplayProof\hspace{0.5cm}
	\AxiomC{$\Pi\phi$}
	\UnaryInfC{$\Pi \exists v(D_v) \phi \wedge (\alpha \rightarrow (v \leftrightarrow \bigvee_{y\in Y} y))$}
\end{prooftree}
The extension variable $v$ is a new variable not appearing in $\Pi$, nor in $\phi$.
$D_v$ is calculated as the union over all $D_{\var(y)}$ for $y\in Y$ and we then subtract the domain of $\alpha$. 
This means $v$ is independent of every variable in $\alpha$ and  even if some variable $x$ in $Y$ does depend on some variable $u$, that dependence will be removed if $u \in \alpha$. 
In the earlier extension rule \cite{Jus07,BBCP20}, the variables that extension variable was defined on coincided with its dependency set. In our new rule $v$ can be defined on variables and not receive its full dependency set. Having a smaller dependency set means that $v$ prevents fewer reduction steps. 
For why this is permissible, we can think of the equational part of the definition only applying once $\alpha$ is already set, therefore there is no scenario of the $\alpha$ variables where $v$ is required to consider a different input value for these variables, other than in the situation where it must consider $\alpha$.


The downside of this definition is that substituting a Boolean function $b$ for its extension variable $v$ adds condition $\alpha$.
But because we subtract the domain of $\alpha$, $v$ no longer blocks the reduction of variables from $\alpha$ so in many cases we can reduce these variables.

\begin{lemma}\label{lem:IndExt}
		Suppose $\Pi \phi$ is a true S-form DQBF, and $\Pi \exists v(D_v) \phi \wedge (\alpha \rightarrow (v \leftrightarrow b(Y)))$ is constructed according to the Independent Extension rule, where $b$ is a Boolean function.
		Then $\Pi \exists v(D_v) \phi \wedge (\alpha \rightarrow (v \leftrightarrow b(Y)))$ is a true S-form DQBF.
\end{lemma}

\begin{proof}
	$\Pi \phi$ is a true S-form DQBF, therefore it has Skolem functions $\sigma_x$ for each existential variable $x$. This set of Skolem functions satisfies all lines in $\phi$, but $\alpha \rightarrow (v \leftrightarrow b(Y))$ may or may not be satisfied.
	To make sure it is satisfied we use a Skolem function $\sigma_v$ for $v$. We apply substitution of $\alpha$ to $Y$ to get $\alpha(Y)$ which assigns some variables to constants. Notice those variables in the domain of $\alpha$ are now constant and no longer  in the domain of $\alpha(Y)$. The free variables of $\alpha(Y)$ are those in $Y$ but not $\alpha$. We define $\sigma_v(D_v)=b(\alpha(Y))$, which works because $b(\alpha(Y))$'s dependency set is $(\bigcup_{y\in Y} D_y) \setminus D_\alpha$.
	Note that $v$ does not appear in $\phi$, so $\phi$ remains satisfied.
	
	If $\alpha$ is not satisfied then  $\alpha \rightarrow (v \leftrightarrow b(Y))$ is automatically satisfied.
	If $\alpha$ is satisfied then $Y=\alpha(Y)$. If $b(Y)$ is true then $b(\alpha(Y))$ is true and so $\sigma_v$ is true which makes $v$ true, so $ v \leftrightarrow b(Y)$ is satisfied.
	Likewise if $b(Y)$ is false, then $b(\alpha(Y))$ is false so $\sigma_v$ and thus $v$ are false, therefore $ v \leftrightarrow b(Y)$ is satisfied.\qed
\end{proof}

	
\subsection{A Sound and Complete Proof System}
	
We define our new proof system \dFregeRed in Figure~\ref{fig:dfred}. The proof system works as a refutational proof typically does, starting with a DBQF $\forall U \exists E \phi$. $\phi$ is a propositional formula. For QBFs, \dFregeRed automatically generalises \eFregeRed, but since we also desire DQBF completeness only p-simulation of a complete DQBF proof system suffices, here we choose \dirc given in Figure~\ref{fig:irc}.
\dirc works by removing all universal literals from the formula and replacing each existential literal with an annotated literal, the annotations are partial universal assignments. The idea is that you can remove universal quantifiers by expansion, but you create multiple copies of the inner existential variables, so the annotations track which expansions have led to this literal.
A p-simulation works because these annotated variables can be defined by independent extension, where they drop the dependence on universal variables that have been expanded on. 

\begin{figure}
	\framebox{\parbox{0.95\textwidth}{
			\begin{minipage}{\linewidth}
				\vspace{0.5em}
				{\centering\textbf{Axiom rule:} \axiom$(\phi)$\vspace{1em}\\}
				\begin{minipage}{.4\linewidth}
					\begin{prooftree}
						\AxiomC{}
						\UnaryInfC{$\{a^{\tau|_{D_a}} \mid a \in C, \var(a) \in \exists\}$}
					\end{prooftree}
				\end{minipage}
				\begin{minipage}{.55\linewidth}
					\begin{itemize}
						\item $C$ is a clause in the matrix $\phi$.
						\item $\tau|_{D_a} := \{\lnot l \mid l \in C, \var(l) \in \forall, \var(l) \in D_a \}$.
					\end{itemize}
				\end{minipage}
			\end{minipage}
			\vspace{0.5em}\\
			\begin{minipage}{\linewidth}
				\vspace{1em}
				{\centering\textbf{Instantiation rule:} \instant$(C,\beta)$\\}
				\begin{minipage}{.4\linewidth}
					\begin{prooftree}
						\AxiomC{$C$}
						\UnaryInfC{$\{ a^{(\alpha \circ \beta)|_{D_a}} \mid a^\alpha \in C \}$}
					\end{prooftree}
				\end{minipage}
				\begin{minipage}{.55\linewidth}
					\begin{itemize}
						\item $\beta$ is a partial assignment to universal variables.
						\item$\alpha \circ \beta$ is $\{l \mid (l\in \alpha) \vee (l\in \beta \wedge \bar l\notin\alpha)\}$.
						\item$(\alpha \circ \beta)|_{D_a}:= \{l \in \alpha \circ \beta \mid \var(l) \in D_a\}$.
					\end{itemize}
				\end{minipage}
			\end{minipage}\\
			\begin{minipage}{\linewidth}
				\vspace{1em}
				{\centering\textbf{Resolution rule:} \res$(C_1,C_2,x^\tau)$\\}
				\begin{minipage}{.4\linewidth}
					\begin{prooftree}
						\AxiomC{$C_1$}
						\AxiomC{$C_2$}
						\BinaryInfC{$(C_1 \cup C_2) \setminus \{x^\tau,\lnot x^\tau\}$}
					\end{prooftree}
				\end{minipage}
				\begin{minipage}{.55\linewidth}
					\begin{itemize}
						\item $x^\tau \in C_1$ and $\lnot x^\tau \in C_2$.
					\end{itemize}
				\end{minipage}
			\end{minipage}
	}}
	\caption{Proof rules of \dirc~\cite{BeyersdorffBlinkhornChewSchmidtSuda19} \label{fig:irc}.}
\end{figure}

\begin{theorem}
	There is an $O(w^2l)$ \dFregeRed p-simulation of \dirc. Where $l$ is the number of lines in $\pi$, $w$ is the size of the largest clause $w=\max_{C\in \pi}(\sum_{l^\alpha\in C}1+ |\alpha|)$ ($\pi$ being the \dirc proof). 
\end{theorem}
\begin{proof}
	\dirc uses existential annotated variables. For each annotated variable $x^\alpha$ appearing in the \dirc proof  we introduce definition clauses $\bar \alpha \vee x^\alpha \vee \bar x$ and $\bar \alpha \vee \bar x^\alpha \vee  x$ based on  $\alpha\rightarrow (x^\alpha\leftrightarrow x)$. We therefore add $O(wl)$ many clauses each of width bounded by $O(w)$.
	
	\noindent\textbf{Axiom:} in \dirc an axiom involves some instantiation of a clause $C$.
	$\tau$ is the partial assignment that contradicts the universal literals in $C$.
	We replace each existential literal $l$ with $l^{\tau|_{D_l}}$ where $\tau|_{D_l}$ restricts $\tau$'s domain to variables in $D_{\var(l)}$. We obtain this by resolving with $\bar\tau|_{D_l} \vee l^{\tau|_{D_l}} \vee \bar l$.
	We accumulate universal literals from $\bar \tau$ in our axiom, but these can be reduced now there are no existential literals that block these literals.
	
	\noindent\textbf{Res:} The resolution step is easy to p-simulate as \Frege p-simulates resolution. 
	
	\noindent\textbf{Inst:} The final rule allows us to instantiate to increase the universal annotation uniformly everywhere in a clause. One $\instantiate(C,\beta)$ of size $O(w)$ can be simulated by lines that total size $O(w^2)$.  
	Instantiation may replace literal $x^\alpha$ with $x^{\alpha\sqcup\beta'}$ for some $\beta'\subseteq \beta$. We already have definitions $\bar \alpha \vee \bar x^\alpha \vee  x$, $\bar \alpha \vee \bar \beta' \vee x^{\alpha\sqcup\beta'} \vee \bar x$. Resolving over $x$ gets us $\bar \alpha \vee \bar \beta' \vee x^{\alpha\sqcup\beta'} \vee \bar x^ \alpha$. Neither $x^{\alpha}$ nor $x^{\alpha\sqcup\beta'}$ have $\alpha$'s variables in its dependency set, so we can now reduce the the literals of $\bar \alpha$ to get $\bar \beta' \vee x^{\alpha\sqcup\beta'} \vee \bar x^ \alpha$.
	
	We use these clauses to instantiate via resolving $O(w)$ many times. $\beta'$ is necessarily a sub assignment of $\beta$, but we may accumulate any of the literals of $\bar \beta$ in our clause.
	Instantiation means that there will be no existential variable remaining that will depend on the variables of $\beta$. And so all universal literals that accumulate can be reduced. 
	The lines are of $O(w)$ size and we involve $O(w)$ many of them to simulate this line.	\qed
\end{proof}

\begin{example}
	Let our DQBF prefix be $\forall u v w \exists a(u,v) b(w)$.	
	Suppose we have instantiation step $\instantiate(a^u \vee b, v\bar w)$.
	In a p-simulated \dirc proof we  already have some clauses that define annotated variables $a^u$, $a^{uv}$, $b^{\bar w}$.
	We can resolve $(\bar u \vee \bar a^u \vee a)$ and $(\bar u \vee \bar v \vee a^{uv} \vee \bar a)$ to get $(\bar u \vee \bar v \vee \bar a^u \vee a^{uv})$, now we reduce $u=1$ to get $(\bar v \vee \bar a^u \vee a^{uv})$.
	Using $(\bar v \vee \bar a^u \vee a^{uv})$ and $a^u \vee b$ we get $\bar v \vee a^{uv} \vee b$, we resolve again with $(w \vee b^{\bar w} \vee \bar b)$ to get $\bar v \vee w \vee a^{uv} \vee b^{\bar w}$. We then reduce with $v=1, w=0$ to get $a^{uv} \vee b^{\bar w}$ which is exactly what $\instantiate(a^u \vee b, v\bar w)$ becomes under our prefix.
\end{example}

Notice that unlike previous simulations of \irc such as the simulation of \irc by \eFregeRed~\cite{chew2022towards}, we are not formalising the strategy, but going line-by-line and replicating each line and its original semantic meaning. In QBF \irc relies on the base propositional inference rule being as weak as resolution for strategy extraction to be possible, but here we can use stronger forms of inference on instantiated clauses.

\begin{corollary}
	\dFregeRed is refutationally complete for S-form DQBFs.
\end{corollary}
\begin{proof}
Technically, simulating \dirc only shows this for DQBF with CNF matrices.
Any propositional formula can be turned into a logically equivalent CNF through enumerating all falsifying assignments to $\phi$. This can be done with \Frege rules and 
the refutation can then proceed by simulating \dirc.\qed
\end{proof}

	\begin{figure}
		\framebox{\parbox{0.95\textwidth}{
				All rules and axioms, for any \Frege system of choice.
				
				\begin{prooftree}
				\AxiomC{}
				\RightLabel{(Axiom)}
				\UnaryInfC{$L$}
				\end{prooftree}
				
				$L$ is a conjunct in the propositional matrix $\phi$.
				
				\begin{prooftree}
					\AxiomC{}
					\RightLabel{(IndExt-$\wedge$)}
					\UnaryInfC{$(\alpha \rightarrow (v \leftrightarrow \bigwedge_{y\in Y}  y))$}
					\DisplayProof\hspace{1cm}
					\AxiomC{}
					\RightLabel{(IndExt-$\vee$)}
					\UnaryInfC{$(\alpha \rightarrow (v \leftrightarrow \bigvee_{y\in Y} y))$}
				\end{prooftree}
				$v$ is a fresh $\exists$ variable, $\alpha$ is a conjunction of $\forall$ literals. $D_v= (\bigcup_{y\in Y} D_y) \setminus D_\alpha$.
				
				\begin{prooftree}
					\AxiomC{$L(u)$}
					\RightLabel{($0$-\textsf{red})}
					\UnaryInfC{$L(0)$}
					\DisplayProof\hspace{1cm}
					\AxiomC{$L(u)$}
					\RightLabel{($1$-\textsf{red})}
					\UnaryInfC{$L(1)$}
				\end{prooftree}
		$u$ is a $\forall$ variable. There is no $\exists$ variable $x$ in $\var(L)$ such that $u\in D_x$.
		\\
		As an additional rule, 
		the prefix $\Pi$ may be weakened to $\Pi'$ to add a new variable that does not appear in the matrix.
		}}
		\caption{Proof rules of \dFregeRed \label{fig:dfred}.}
	\end{figure}
	
	
\begin{theorem}
	\dFregeRed is a sound refutational proof system. 
\end{theorem}

\begin{proof}
	We claim that if DQBF $\forall U \exists E \phi$ is true and a number of lines $L_1 \dots L_n$ are derived by \dFregeRed from $\forall U \exists E \phi$ , then   $\forall U' \exists E' \phi\wedge L_1\wedge  \dots \wedge L_n$ is also a true DQBF, where $E'$ and $U'$ extend the prefix only to include the new variables added by prefix weakening or the Independent Extension rules from $L_1 \dots L_n$.
	We can prove that if $\forall U' \exists E' \phi\wedge L_1\wedge  \dots \wedge L_{n}$ has a set of Skolem functions that satisfies all conjuncts then $\forall U'' \exists E'' \phi\wedge L_1\wedge  \dots \wedge L_{n+1}$ has a set of Skolem functions that satisfies all conjuncts (where $E''$ and $U''$ extend $E'$ and $U'$, respectively to include the new variables of $L_{n=1}$). 
	Axiom and Frege rules and reduction rules preserve winning Skolem functions.  
	For Axiom and Frege this is easy to see. Since these rules preserve models in propositional logic they preserve the whether a Skolem function satisfies all lines.
	For reduction, if every line is satisfied by the Skolem functions, this includes $L(u)$, and this does not change under both values of $u$ and factoring in these values does not necessitate updating any of the Skolem functions of the variables of $L$ because they do not depend on $u$ (Lemma~\ref{lem:redsound}).
	For the IndExt axioms we know they preserve DBQF truth from Lemma~\ref{lem:IndExt}.\qed
\end{proof}
	


In the following example, we demonstrate that \dFregeRed is conditionally strictly stronger than \eFregeRed. 

\begin{example}
	Let $X$ be the set of variables $\{x_1, \dots, x_{2n} \}$ and $\phi(X)$ a CNF in the variables of $X$.
	Then $\Pi \phi(X)$ with prefix $\Pi= \forall x_1 \exists x_2 \forall x_3 \dots \exists x_{2n}$ is a closed PCNF.
	We also define a second set of mirrored variables $X'= \{x'_1, \dots, x'_{2n}\}$.
	We can define two false QBF families:
	
	$\duality(\Pi\phi) = \exists x'_1 \forall x_1 \exists x_2 \forall x'_2 \dots \exists x'_{2n-1} \forall x_{2n-1} \exists x_{2n} \forall x'_{2n} \phi(X) \wedge \neg\phi(X')$,
	
	$\select(\Pi\phi) = \forall u \exists x'_1 \forall x_1 \exists x_2 \forall x'_2 \dots \exists x_{2n} \forall x'_{2n} (\phi(X) \vee u) \wedge (\neg\phi(X') \vee \neg u).$
	
	 Due to the prefix ordering, $\duality$ has an easy strategy for each variable, but $\select$ has a $\PSPACE$-hard strategy for its first variable $u$. It was shown \cite{ChewCly20} that 
	using the easy strategy $\duality$ always has a short proof in \eFregeRed.
	Because of strategy extraction, $\select$ cannot have short proofs in \eFregeRed unless $\PSPACE\subset \PeP$, but they have short proofs in \dFregeRed. 
	
	For each variable $x_i$ we create new variable $y_i$, if $x_i$ is existential in $\Pi$ then we use definition $\neg u \rightarrow (y_i\leftrightarrow x_i)$ if $x_i$ is universal then we use definition $u \rightarrow (y_i\leftrightarrow x_i')$. 
	These variables have a dependency subset of $x_i$ and $x_i'$, respectively. Importantly, they are independent of $u$, much like in the $\duality$ formula where there is no $u$ variable. 
	Now we replace all existential variables $x_i$ and $x'_i$  in $(\phi(X) \vee u) \wedge (\neg\phi(X') \vee \neg u)$ with $y_i$ variables. This gives us two conjuncts $(\phi(Y) \vee u)$ and $(\neg\phi(Y') \vee \neg u)$ where $Y=\{y_i\mid  \exists x_i\in \Pi\}\cup \{x_i \mid \forall x_i\in \Pi\}$ and  $Y'=\{y_i\mid  \exists x'_i\in \Pi'\}\cup \{x_i'\mid \forall x'_i\in \Pi'\}$.
	The conditional part of the definition is absorbed by the $\vee u$ and $\vee \neg u$ part, respectively.
	In each conjunct, $u$ can reduced due to independence from all the existential variables. Therefore we have $\phi(Y) \wedge \neg\phi(Y')$, and structurally we have the $\duality(\Pi\phi)$ formula, so we simply proceed with the short \eFregeRed proof to get a short proof. Because these proofs are short and uniform, \dFregeRed cannot have polynomial time strategy extraction unless $\Pe=\PSPACE$ and so \eFregeRed cannot p-simulate \dFregeRed unless $\Pe=\PSPACE$.

\end{example}

\subsection{A Resolution Version}\label{sec:dres}
	We can better demonstrate the simplicity of the new rule by defining an equivalent resolution system that uses only four clausal rules. We give the proof system in Figure~\ref{fig:dres}.
	\begin{figure}
	\framebox{\parbox{0.95\textwidth}{

			\begin{prooftree}
				\AxiomC{}
				\RightLabel{(Ax)}
				\UnaryInfC{$L$}
				\DisplayProof\hspace{1cm}
				\AxiomC{$C \vee u$}
				\RightLabel{(Red)}
				\UnaryInfC{$C$}
				\DisplayProof\hspace{1cm}
				\AxiomC{$E \vee \neg x$}
				\AxiomC{$F \vee x$}
				\RightLabel{(Res)}
				\BinaryInfC{$E \vee F$}
			\end{prooftree}
			
			$L$ is a clause in the propositional matrix $\phi$. 
			$u$ is a $\forall$ literal. There is no $\exists$ literal $l$ in $C$ such that $\var(u)\in D_{\var(l)}$, and there is no $\bar u\in C$.
			
			\begin{prooftree}
				\AxiomC{}
				\RightLabel{(IndExt)}
				\UnaryInfC{$(\bar\alpha \vee v \vee  y_1),(\bar\alpha \vee \bar v \vee   y_2),(\bar\alpha \vee  \bar v \vee \bar y_1 \vee \bar y_2)  $}
			\end{prooftree}
			$v$ is a fresh $\exists$ variable, $\alpha$ is a conjunction of $\forall$ literals. $D_v= (D_{y_1} \vee D_{y_2} ) \setminus D_\alpha$.
			\medskip\\
			As an additional rule, 
			the prefix $\Pi$ may be weakened to $\Pi'$ to add a new variable that does not appear in the matrix.
	}}
	\caption{Proof rules of \dRes \label{fig:dres}.}
\end{figure}

\begin{theorem}
	\dRes and \dFregeRed are p-equivalent.
\end{theorem}

\begin{proof}
	
	($\Leftarrow$) We p-simulate each individual rule. (Ax) in \dRes is a straightforward applications of (Ax) in \dFregeRed. (Red) in \dRes can be p-simulated using $(0-\mathsf{red})$ or $(1-\mathsf{red})$, depending on whether the reduced literal was positive of negative, respectively. Using \Frege we can say a literal equal to $0$ or $\neg 1$ is equivalent to it being removed. It is well-known that the resolution rule is p-simulated by \Frege. For the (IndExt) rule we us the (IndExt-$\vee$) rule in Figure~\ref{fig:dfred} on $Y=\{\bar y_1, \bar y_2\}$. Then each of the three clauses for (IndExt) in Figure~\ref{fig:dres} is a propositional implicant of the formula derived by (IndExt-$\vee$) from Figure~\ref{fig:dfred} and can be derived using a short \Frege proof. Finally whenever the empty clause $\bot$ is derived in \dRes, the final steps either use a resolution step or a reduction step. In either case we can p-simulate and derive $0$ instead. 
	
	$(\Rightarrow)$ we interpret every line $L_i$ in \dFregeRed as an extension variable $l_i$ that is built as the circuit for the formula in the line. 
	All Frege rules and axioms can be p-simulated by the well know p-simulation of Ext Frege by Ext Res. The axiom rule for \dFregeRed will technically have to be p-simulated by deriving a singleton clause with the variable $l$ representing the disjunction, but every literal $x$ in the axiom clause will be used in a definition clause $\neg x \vee l$. and so we resolve away the literals until we get singleton $l$.
	
	For (IndExt-$\wedge$) and (IndExt-$\vee$), consider Figure~\ref{fig:dfred} where we define a new variable $v$ we represent the $\bigwedge_{y\in Y} y$ and $\bigvee_{y\in Y} y$ formulas from Figure~\ref{fig:dfred} with an extension variable $p$, we take the variables $y_1$ and $y_2$ from Figure~\ref{fig:dres} to both be $\bar p$ and then we define $v$ using (IndExt) (using the same $\alpha$). Once we define extension variable $q$ with $l\leftrightarrow (\alpha\rightarrow (v\leftrightarrow w))$ we can resolve to get singleton $l$. 
	
	Suppose we perform a reduction from $L(u)$ to $L(0)$ in \dFregeRed. Let us label $L(u)$ as $p$ and $L(0)$ as $q$. $L(u)\wedge \bar u \rightarrow L(0)$ is an obvious propositional tautology and we use the p-simulation of Ext Frege by Ext Res to derive clause $\bar p \vee u \vee q$, as we do not have weakening in \dRes, we may obtain a stronger clause, which is just as useful. Once we obtain singleton $p$ we resolve it to get $q \vee u$. $u$ is not blocked by $q$ and so we reduce to get $q$, or even the empty clause (which saves us from all subsequent lines). This works symmetrically for $(1-\mathsf{red})$ as it does for $(0-\mathsf{red})$
	
	And finally if we derive $l$ which represents the empty clause from \dFregeRed, we simply resolve with $\neg l$ which is part of the definition of $l$ to get the empty clause in \dRes.
	\qed
\end{proof}


\section{P-simulations}\label{sec:psim}

In order to show completeness we have already shown that \dFregeRed p-simulates \dirc. We can show more p-simulations that demonstrate the power of this new extension rule, including Fork Resolution, QRAT,
\Gfull, 
 and \Drrs based systems. With these
p-simulations we show that \dFregeRed can capture the vast
majority of (D)QBF solving and preprocessing techniques.

These p-simulations also demonstrate that \dFregeRed is powerful. A proof systems is roughly as powerful as the most expressive object it can ``cut''. Resolution cuts on literals, \Frege cuts on propositional formulas, bounded-depth \Frege cuts on bounded depth formulas and $\Gfull$ cuts on QBFs.  The extension variables in \eFrege and Ext Res cut on propositional circuits, as is the case in \eFregeRed.
\dFregeRed allows extension variables to be more expressive and by combining them we can express $\PSPACE$-hard objects, and we make use of this in the p-simulations of this section. 


	\subsection{P-Simulation of  Fork Resolution}\label{sec:fork}
	In this section we show how \dFregeRed can p-simulate another DQBF proof system that has a different notion of extension variables.

The Fork Resolution proof system is sound and complete for DQBFs that have a CNF matrix. Its main novelty is a Fork Extension rule~\cite{rabe17}, which is used for splitting clauses.
\begin{prooftree}
	\AxiomC{$C_1\cup C_2$}
	\RightLabel{(F-Ext)}
	\UnaryInfC{$ e \vee C_1$ \quad $\bar e \vee C_2$}
	\DisplayProof\hspace{1cm}
	\AxiomC{$ x \vee C_1$}
	\AxiomC{$\bar x \vee C_2$}
	\RightLabel{(Res)}
	\BinaryInfC{$C_1\cup C_2$}
	\DisplayProof\hspace{1cm}
	\AxiomC{$C\vee u$}
	\RightLabel{($\forall$\textsf{red})}
	\UnaryInfC{$C$}
\end{prooftree}
$e$ is a new $\exists$-variable and has dependency set $(\bigcup_{\var(y)\in C_1} D_y)\cap (\bigcup_{\var(y)\in C_2} D_y)$. 
Fork Resolution also uses a resolution rule and a reduction rule:
provided $\bar u$ is not in $C$ and $\var(u)$ is not in the dependency set of $C$. 


Given an instance of the Fork Extension on $C_1 \cup C_2$, the idea is to get an extension variable $x$ that is equivalent to $\forall \{u\in U \mid u\in (\bigcup_{y\in C_2} D_y), u \notin \bigcup_{y\in C_1} D_y)  \} C_2$. In other words $x$ is true if $C_2$ is true under all assignments to the specific universal variables that govern $C_2$ but not $C_1$.

It should be clear why this should  mean $C_1\cup C_2 \rightarrow \bar x \vee C_2$, because $x$ is a stronger version of $C_2$. $C_1\cup C_2 $ also implies $x\vee C_1$ is true because if a Skolem function always satisfies $C_1\cup C_2 $, then $C_2$ must be satisfied, whenever $C_1$ is falsified regardless of the values of any remaining dependencies of $C_2$.

\begin{lemma}
	\dFregeRed can p-simulate the Fork Extension rule (F-Ext).
\end{lemma}

\begin{proof}
	We order all universal variables in the \emph{marginal} dependency set of $C_2$, $\{u_1, \dots ,u_k\}=\{u\in U \mid u\in (\bigcup_{y\in C_2} D_y), u \notin (\bigcup_{y\in C_1} D_y)  \}$. We define:
	$$e_0\leftrightarrow C_2, \enspace u_{i+1}\rightarrow (e_i^{u_{i+1}}\leftrightarrow e_i), \enspace \bar u_{i+1}\rightarrow (x_i^{\bar u_{i+1}}\leftrightarrow e_i), \enspace e_{i+1}\leftrightarrow (e_i^{u_{i+1}} \wedge e_i^{\bar u_{i+1}})$$
	And we can use the extension clauses that give these definitions.
	We make the induction hypothesis that there is a short proof of $e_i \vee C_1$ and of $\bar e_i \vee C_2$ in \dFregeRed. Starting from $i=0$ and incrementing until $i=k$.
	
	\noindent\textbf{Base Case:} We resolve off all the $C_2$ literals in $C_1 \vee C_2$ to get $C_1 \vee e_0$. $\bar e_0 \vee C_2$ is part of the definition of $e_0=C_2$.
	
	\noindent\textbf{Inductive Step:} We start with  $C_1 \vee e_i$ and $\bar e_i \vee C_2$. 
	Resolving $\bar u_{i+1} \vee e_i \vee \bar e_i^{u_{i+1}}$ with $u_{i+1} \vee e_i \vee \bar e_i^{\bar u_{i+1}}$  gets us $e_i \vee \bar e_i^{u_{i+1}}\vee \bar e_i^{\bar u_{i+1}}$ and then with $\bar e_{i+1} \vee e_{i}^{u_{i+1}}$ and $\bar e_{i+1} \vee e_{i}^{\bar u_{i+1}}$, we get that $ e_i \vee \bar x_{i+1}$. 
	We can then, with a resolution step, get $\bar e_{i+1} \vee C_2$
	
	In $C_1 \vee e_i$ we replace literal $e_i$ with $\bar u_{i+1} \vee e_i^{u_{i+1}}$. Since $\bar u_{i+1}$ is not blocked by any literal in $C_1$ and not blocked by $e_i^{u_{i+1}}$ it can be reduced to derive $C_1 \vee e_i^{u_{i+1}}$, $C_1 \vee x_i^{\bar u_{i+1}}$ is derived in the same way. We resolve both with $ e_{i+1}\vee \bar e_i^{u_{i+1}} \vee \bar e_i^{\bar u_{i+1}}$ to get $C_1 \vee e_{i+1}$
	
	This DAG-like process is linear in the number of induction steps, potentially multiplied by the width of $C_1 \cup C_2$.  
	To finalise, we take $e=e_k$, and necessarily it is not dependent on variables not in the dependency set of $C_2$ because of the initial definition of $e_0$ and any variable not in the dependency set of $C_1$ is removed by the time we get to $e_k$. \qed
	\end{proof}

\begin{corollary}
	\dFregeRed p-simulates Fork Resolution.
\end{corollary}

Unlike in \eFregeRed where new variables are essentially propositional circuits, here we show how a new variable can be efficiently constructed to resemble a QBF properly. \dFregeRed shares the ability to cut over a QBF with \Gfull (as we will see in more detail in Section~\ref{sec:G}), which goes some way to explain why \dFregeRed is so powerful. 

\subsection{P-simulation of QRAT}\label{sec:qrat}
In this section with show our DQBF proof system \dFregeRed can p-simulate the QBF proof system QRAT\cite{HSB14}, which governs many preprocessing steps in QBF.
\subsubsection{Definition of QRAT}

In propositional logic, interference based proof systems modify a CNF into another equisatisfiable CNF without preserving models. Resolution, is not an interference based proof system because it preserves logical equivalence.
An example of an interference based proof system is DRAT \cite{WetzlerHH14}, where clause addition can destroy an existing model, but only when there is a guarantee that another model exists and is preserved.
QRAT has similar rules to DRAT, but since QRAT is a QBF proof system, quantifier order must be respected. Since QRAT works in QBF, the notion of preserving individual models, is understood as preserving Skolem functions. QRAT does not preserve Skolem functions.

\begin{definition}[\cite{Blinkhorn2020SimulatingDP}]\label{def:O}  Fix a QBF prefix $\Pi$. Now consider a clause $D$ and a literal $l$ (not necessarily in $D$) we define,
	$O_D=\{k\in D \mid \lv(k)\leq \lv(l), \var(k)\neq \var(l)\}$, where $\lv(k)\leq \lv(l)$ indicates that $k$ has a lower quantifier level than $l$, i.e $k$ is left of $l$ or, because we include equality, both are part of the same quantifier block (i.e. same dependency set).  
	$O_D$ is called the  \emph{outer clause}. 
	For a DQBF we define the outer clause $O_D$ differently for existential and universal literals $l$. 
	
	When $l$ is existential:
	$O_D:=\{x\in D \mid x\in \exists, D_{\var(x)}\subseteq D_{\var(l)}, \var(x)\neq \var(l)\}\cup \{x\in D \mid x\in \forall, \var(x)\in
	 D_{\var(l)}\}$
	 
	 When $l$ is universal we define a set of existential variables $\mathcal{S}=\{x\in \exists \mid \var(l)\in D_x\}$ from that we define a set of universal variables $T=\bigcap_{x\in S} D_x \setminus\{\var(l)\}$, and we can then define the outer clause of $D$:
	 $O_D:=\{x\in D \mid x\in \exists, D_{\var(x)}\subseteq T\}\cup \{x\in D \mid x\in \forall, \var(x)\in T\}$.
\end{definition}

The outer clause is a concept used to evaluate whether new clauses can be added, as it is the critical part that decides soundness. This is in the same vein as how extension clauses rely on the outer variables of the definition.

For some rules in QRAT, they can only be applied soundly if some semantic implication is known, but for a proof system we need to know if these implications are true in polynomial time. To do this we use the sound but incomplete unit propagation procedure. 

\begin{definition}
	$\phi \vdash_1 \bot$ means we arrive at the empty clause through the following process:
	For every clause in $\phi$ of size one (in other words a unit clause), $\{l\}$ we add $l$ to a partial assignment $\alpha$  ($\var(l)$ is assigned the polarity of literal $l$). Now we remove literals from clauses in $\phi$ whenever they conflict with $\alpha$. 
	This may produce more unit clauses, or even an empty clause, we repeat this process until no more unit clauses are produced. 
\end{definition}

When it comes to simulate QRAT, fortunately unit propagation is something that can be easy to p-simulate in a sequence of rules.
QRAT's first rule uses unit propagation very directly to infer clauses propositionally, in this rule we do not have to consider quantifier type or order.

\begin{definition}[Asymmetric Tautology Addition (ATA)]\label{def:ATA}
	Let $\phi$ be a CNF with $\Pi$ a prefix.
	Let $C$ be a clause not in $\phi $. 
	Let $\Pi'$ be a prefix including the variables of $C$ and $\phi$, $\Pi\subset \Pi'$.
	
	Suppose $ \phi\wedge  \bar{C}\vdash_{1} \bot $. 
	Then we can infer $\Pi' \phi \wedge C$ from $\Pi \phi$.
	
\end{definition}

The next rule, QRATA deals with adding or removing a clause, but this time the Skolem function for a particular existential literal $l$ changes as a result of this rule. This means that QRATA preserves truth but does not necessarily preserve the strategies. QRATA is a generalisation of the extension rule in \eFregeRed, and as such we have to respect the quantification order.

\begin{definition}[Quantified Resolution Asymmetric Tautology Addition (QRATA)]\label{def:QRATA}
	Let $\Pi \phi$ be a PCNF with closed prefix $\Pi$ and  CNF matrix $\phi$. Let $C$ be a clause not in $\phi $. Let $\Pi_1$ and $\Pi_2$ be disjoint  prefixes and $x$ a variable such that  $\Pi \subseteq \Pi_1\exists x\Pi_2$ and a literal $l$ with $\var(l)=x$. The difference in prefix is simply to allow new variables coming from $C\vee l$.
	For every clause $D \in \phi$ with $\bar l\in D$ if
	$ \phi\wedge  \bar{C}\wedge \bar{l} \wedge \bar{O}_D\vdash_{1} \bot $,
	then we can derive $\Pi_1\exists x \Pi_2 \phi \wedge (C \vee l)$ from $\Pi \phi$.
%
%
	
\end{definition}

The next two rules QRATU and EUR remove a universal literal from a clause, with side conditions we will soon introduce.
	\begin{prooftree}
	\AxiomC{$\Pi_1\forall u\Pi_2 \phi \wedge (C\vee l) $}
	\RightLabel{(QRATU)}
	\UnaryInfC{$\Pi_1\forall u\Pi_2 \phi \wedge C$}
	\DisplayProof\hspace{2cm}
	\AxiomC{$ \Pi_1 \forall u \Pi_2  \phi \wedge (C \vee l)$}
	\RightLabel{(EUR)}
	\UnaryInfC{$\Pi_1 \forall u \Pi_2 (\phi \wedge C)$}
	\end{prooftree}	
For QRATU the condition is similar to that of QRATA but instead of adding a blocked clause over an existential literal, it removes a blocked universal literal. EUR has a different condition based on resolution paths.

\begin{definition}[Quantified Resolution Asymmetric Tautology 
	Universal (QRATU)]\label{def:QRATU}
	Let $\Pi_1\forall u\Pi_2\phi$ be a PCNF with closed prefix $\Pi_1 \forall u \Pi_2$ and  CNF matrix $\phi$. Let $C\vee l$ be a clause with universal literal $l$, with $\var(l)= u$. 
	If for every $D \in \phi$ with $\bar{l} \in D$, 
	$\phi\wedge  \bar{C}\wedge \bar{O}_D\vdash_{1} \bot $,
	then we can derive $\Pi_1\forall u\Pi_2 \phi \wedge C$ from $\Pi_1\forall u\Pi_2 \phi \wedge (C\vee l) $.

\end{definition}

For EUR, we consider potential resolution steps connecting clauses to one another. Clauses being in different connected components indicate independence from one another and we can expand on this idea to calculate when a universal literal is locally pure and can only be considered one polarity.  

\begin{definition}\label{def:rpath}
	
	Consider a CNF $\phi$ and subset $\chi$ of clauses in $\phi$ and a subset $\mathcal{S}$  of variables.  $\pathl(\phi, \chi, \mathcal{S})$ lists the $\mathcal{S}$-\emph{literals on the resolutions paths from $\chi$} and $\pathc(\phi, \chi, \mathcal{S})$ lists \emph{the clauses on the the resolution paths from $\chi$}.
	These are found using an iterative procedure until reaching a fix-point. 
	
	\noindent\textbf{Initialisation.}
	We start with the clauses in $\chi$ and the $\mathcal{S}$ literals in those clauses. $\pathl(\phi, \chi, \mathcal{S}) \gets \{l \mid \text{there is } C\in\chi \text{ s.t. } l \in C, \var(l)\in \mathcal{S} \}$ and $\pathc(\phi, \chi, \mathcal{S})  \gets \chi $.
	
	\noindent\textbf{Adding a clause.}
	If there if some $D\in \phi$ such that $\bar {p}\in D$ and $p\in \pathl(\phi, \chi, \mathcal{S}) $, then we can update $\pathl(\phi, \chi, \mathcal{S}) $ and $\pathc(\phi, \chi, \mathcal{S}) $. $\pathl(\phi, \chi, \mathcal{S}) \gets \pathl(\phi, \chi, \mathcal{S}) \cup \{q\in D \mid q\neq \bar{p}, \var(q)\in \mathcal{S}\}$ and $\pathc(\phi, \chi, \mathcal{S})  \gets \pathc(\phi, \chi, \mathcal{S})  \cup \{ D\}$. 	
	We continue this until we reach fix-point, in other words for all $p \in \pathl(\phi, \chi, \mathcal{S}) $ if $ D\in \phi$ and $\bar p \in D$, then $\{q\in D \mid q\neq \bar{p} ,\var(q)\in \mathcal{S}\} \subset \pathl(\phi, \chi, \mathcal{S}) $ and $D\in \pathc(\phi, \chi, \mathcal{S}) $. Fix-point is reached in polynomial time.
	
\end{definition}

\begin{definition}\label{def:EUR}
	Let $\Pi_1\forall u\Pi_2\phi$ be a PCNF with closed prefix $\Pi_1 \forall u \Pi_2$ and  CNF matrix $\phi$. Let $C\vee l$ be a clause with universal literal $l$, with $\var(l)= u$.

	If the resolution path $\pathc(\phi\wedge C, C, \mathcal{S})$  contains no clause $D$ such that $\bar{l} \in D$, when $\mathcal{S}$ is the set of existential variables right of $l$ in the prefix (i.e. in $\Pi_2$), then we can derive $\Pi_1 \forall u \Pi_2 (\phi \wedge C)$ from $ \Pi_1 \forall u \Pi_2  \phi \wedge (C \vee l)$.

\end{definition}

The refutational system QRAT includes ATA, QRATA, QRATU, EUR rules as well as arbitrary clause deletion and that the prefix allows for new variables to be added.

\subsubsection{Interference-based Reasoning via Independent Extension Variables}

Our p-simulation is somewhat surprising, as QRAT includes a complicated Extended Universal Reduction rule that performs reductions according to global resolution path connectivity between clauses. On the surface level \dFregeRed would seem to be ill-equipped to talk globally about a formula, as it does not have any global conditions. In fact \dFregeRed is monotonic in the sense that no deletion rule is necessary. The way we p-simulate clause deletion is simply to ignore clauses that would have been in the QRAT proof. This has a small technicality, QRAT can eliminate a variable and then introduce it again with a different meaning, but in \dFregeRed we have to take care to create new variables when we refresh the meanings of variables. In fact when we use an interference based rule such as QRATA, QRATU or EUR we create a new copy of the formula.  

We begin with p-simulations of DQBF generalisations of ATA and QRATA, in fact we do not need the independence power of the new Independent Extension rule for this. 
\begin{lemma}[\cite{kiesl2018extended}]\label{lem:ATA}
	Suppose we have a CNF $\phi$ and a clause $C \notin \phi$. If $\phi \wedge \bar C \vdash_1 \bot $ then there is a short resolution + weakening  proof  starting from $\phi$ to $\phi\wedge C$.
\end{lemma}

\begin{proof}
	Any simplification of a clause can be derived via resolution. If we have a unit clause $x$ and clause $C\vee \bar x$, then the clause $C$ can be derived in a single step.
	The number of literals removed before the empty clause is bounded above by the total number of individual literals. Hence the short proof.\qed
\end{proof}
%

\begin{lemma}\label{lem:QRATA}
	Suppose we have a DQBF: $\forall U \exists E \phi$.
	Let $l$ be a literal of some existential variable in $E$ and $C$ be a clause not in $\phi$, where $C$ contains literal $l$.
	
	With respect to existential literal/variable $l$, recall the DQBF definition of an outer clause $O_K$ of some clause $K$, according to Definition~\ref{def:O}.
	If for every clause $D\in \phi$ with $\bar l \in D$ satisfies:
	$\phi \wedge \bar C \wedge \bar O_D \wedge \bar l \vdash_1 \bot$,
	then using a polynomial bounded (in the size of $\forall U \exists E \psi$) number of rules of \dFregeRed we can extend $\exists E$ to $\exists E'$, which includes a new variable $l'$ such that  $D_{l'}\subseteq D_{\var(l)}$ and derive the clauses of $(\phi \wedge C)'$ which is a copy of $\phi\wedge C$ where every $l$ literal in $\mathcal{S}$ is replaced by $l'$ and every $\bar l$ literal in $\mathcal{S}$ is replaced by $\bar l'$.
\end{lemma}

\begin{proof}
	The aim is to replace $l$ with the substitution $l'= l \vee \bigwedge_{D\in \phi}^{\bar l \in D} O_D$ and this will respect all existing clauses and add $C'$. $l'$ has the same dependency set as $l$ because of the care taken when selecting literals for the outer clause. We can prove the clauses of $\phi'$ by cases:
	
	\begin{enumerate}
		\item Clauses $K\in \phi$ that do not contain $l$ nor $\bar l$ can remain unchanged  in $\phi'$.
		\item Clauses $K\in \phi$ that contain $l$ but not $\bar l$ can be weakened to $ K\vee \bigwedge_{D\in \phi}^{\bar l \in D} O_D$ and the replacement of $l\vee \bigwedge_{D\in \phi}^{\bar l \in D} O_D$ with $l'$ gains us $K'= K\setminus\{l\}\cup\{l'\}$.
		\item  Clauses $K\in \phi$ that contain $\bar l$ but not $l$ have the property that the tautology $ O_K \vee \bar O_K$ implies $K\setminus\{\bar l \} \vee \bigvee_{D\in \phi}^{\bar l \in D} \bar O_D$. Combine with a conjunction of $K$ we get $K\setminus\{\bar l \} \vee (\bar l \wedge \bigvee_{D\in \phi}^{\bar l \in D} \bar O_D)$ and now  can replace $(\bar l \wedge \bigvee_{D\in \phi}^{\bar l \in D} \bar O_D)$ with $\bar l'$ to get $K'= (K\setminus\{\bar l\})\cup\{\bar l'\}$.
		\item Clauses $K\in \phi$ that contain both $l$ and $\bar l$ then  since $l' \vee \bar l'$ is also tautological we can derive $K'= K\setminus\{l, \bar l\}\cup\{l', \bar l'\}$.
	\end{enumerate}
	
	We use the short proof of
	$\phi\rightarrow C \vee l\vee \bigwedge_{D \in \phi}^{\bar u\in D} O_D$ from unit propagation to derive $C \vee l\vee \bigwedge_{D \in \phi}^{\bar u\in D} O_D$ and then replace $l\vee \bigwedge_{D \in \phi}^{\bar u\in D} O_D$ with $l'$ to get $C'\vee l'$.\qed
\end{proof}

Instead of p-simulating EUR and QRATU separately, what we will show here is a simulation of a powerful mixed reduction rule that acts like a combination of QRATU and EUR and even works for DQBF. QRATU and EUR will both be special cases when the prefix gives a QBF. 

\begin{lemma}\label{lem:mixedred}
	Suppose we have a DQBF: $\forall U \exists E \psi$.	
	Let $u$ be a literal of some universal variable in $U$ and $C \vee u$ be a clause in $\psi$, i.e $\psi=\phi\wedge (C\vee u)$.
	Recall Definition~\ref{def:O} for the definition of $\mathcal{S}$
	Let $\chi_0=\pathc(\psi, C\vee u, \mathcal{S})$, the set of clauses reachable from $C\vee u$ via resolution path in $\mathcal{S}$. 
	
	If for every clause $D\in \chi_0$ with $\bar u \in D$ satisfies:
	$\phi \wedge \bar C \wedge \bar O_D \vdash_1 \bot$,
	then using a polynomial bounded (in the size of $\forall U \exists E \psi$) number of rules of \dFrege we can extend $\exists E$ to $\exists E'$, which includes for every variable $x$ in $\mathcal{S}$ a variable $x'$ such that  $D_{x'}\subseteq D_x$ and derive the clauses of $(\phi \wedge C)'$ which is a copy of $\phi\wedge C$ where every $x$ in $\mathcal{S}$ is replaced by $x'$.
	
\end{lemma}

As specified in the lemma the proof is done by replacing literals with different extension variables. The origins of the replacement scheme used here trace back to the simulation of QRAT by \Gfull~\cite{CH22}, although this simulation here simulates a more powerful rule and works in the DQBF setting. This is despite being seemingly handicapped by only being able to create definitions that are existential variables, whereas \Gfull can use complicated functions to witness universal variables. To overcome this we hide functions that previously witness universal variables into new definitions of the existential variables that depend on them. 


\begin{proof}
	We start by introducing the extensions $\bar u \rightarrow (x^{\bar u} \leftrightarrow x)$ and $u \rightarrow (x^u \leftrightarrow x)$ for every $x\in \mathcal{S}$.
	Let $\mathcal{O} = \bigvee_{D\in \chi_0}^{\bar u \in D} \bar O_D$.
	For each literal in $L=\pathl(\psi, C\vee u, \mathcal{S})$ and their complement we now define
	$l^* = \left( \bar u \lor \bar{\mathcal{O}} \rightarrow l^{\bar u} \right) \land \left( u \land \mathcal{O} \rightarrow l^u \right)$.
	In other words, $l^{*}$ is $l^{\bar u}$ except in the circumstance that $u$ is true and some outer clause is falsified, and in that case and only that case we have $l^{*}$ takes the value from $l^{ u}$.
	The definition is consistent with negation (i.e. $\neg l^{*}= (\neg l)^*$). The dependency set of $l^*$ is no greater than that of $l$.
	
	One may notice that if there were Skolem functions that satisfied $\forall U \exists E \psi$, these Skolem functions extended to the new $*$-variables would also satisfy $\chi_0^{*}$, which replaces all $x$ variables with $x^{*}$ in $\chi_0$. This is because if $u$ is false every $l^{*}$ plays as $l^{\bar u}$ which is how $l$ plays anyway. 
	If $u$ is true and some $O_D$ for ${D\in \chi_0}, {\bar u \in D}$ is falsified, all $l^{*}$ variables play consistently as $l^u$ which is consistent with $l$ and $u$ and thus satisfies all clauses.
	If $u$ is true and every $O_D$ for ${D\in \chi_0}, {\bar u \in D}$ is satisfied, $l^{*}$ plays as $l^{\bar u}$ 
	which will still satisfy every clause in $\chi_0$ except these $D$, but each of them is already satisfied from their outer clause. 
	
	We have to actually construct $\chi_0^{*}$ using rules from \dFrege. There are four types of clauses $K$ in $\chi_0^{*}$. For each we have to construct $K^*$ which replaces every $S$-literal $l$ with $l^*$ in $K$: \textbf{1. }$K$ contains no $u$ literal. \textbf{2. }$K$ contains a positive $u$ literal only. \textbf{3. }$K$ that contains a negative $u$ literal only. \textbf{4. }$K$ contains a tautology over $u$.
	
	\begin{enumerate}
		\item For each $\mathcal{S}$-literal  $l$ of $K$ we can use definition $\bar u\rightarrow(l^{\bar u}\leftrightarrow l)$.
		 For $K$ we replace each of these literals $l$ for $l^{\bar u} \vee u$ via resolving with the definitions. At the end of this process we can denote this clause as $K^{\bar u} \vee u$. There is no literal that blocks $u$ from reduction here so we can derive $K^{\bar u}$. 
		We do this similarly with definitions $ u\rightarrow(l^{u}\leftrightarrow l)$ to create $K^{ u}$.
		$( u \wedge \mathcal{O}) \vee K^*\vee \bar K^{\bar u} $ is  a tautology once the definition of $l^*$ is considered which can be proven in a linear size \Frege proof.
		Likewise 	$\bar u \vee \bar{\mathcal{O}} \vee K^*\vee \bar K^{ u} $ is a tautology which can be proven in a linear size \Frege proof.
		And by resolving disjuncts together we are left with $K^*$.
		\begin{prooftree}
			\AxiomC{$K^{\bar u} \vee u$}
			\UnaryInfC{$K^{\bar u}$}
			\AxiomC{$(u \wedge {\mathcal{O}})\vee K^*\vee \bar K^{\bar u}$}
			\BinaryInfC{$(u \wedge {\mathcal{O}})\vee K^*$}
			\AxiomC{$K^{u} \vee \bar u$}
			\UnaryInfC{$K^{u}$}
			\AxiomC{$\bar u \vee \bar{\mathcal{O}}\vee K^*\vee \bar K^{ u}$}
			\BinaryInfC{$\bar u \vee \bar{\mathcal{O}}\vee K^*$}
			\BinaryInfC{$K^*$}
		\end{prooftree}
		
		\item 
		For $K$, we replace each of these literals $l$ for $l^{\bar u} \vee u$. At the end of this process we can denote this clause as $K^{\bar u}$, we do not have to include the extra $u$ literal as that is already included.
		We do this similarly with definitions $ u\rightarrow(l^{u}\leftrightarrow l)$ to create $K^{ u}\vee \bar u$. The $\bar u$ need not be reduced as it will be absorbed into the $\bar u$ disjunct later.
		$( u \wedge {\mathcal{O}}) \vee K^*\vee \bar K^{\bar u} $ is  a tautology once the definition of $l^*$ is considered which can be proven in a linear size \Frege proof.
		Likewise 	$( u \vee \bar{\mathcal{O}}) \vee K^*\vee \bar K^{ u} $  is  a tautology which can be proven in a linear size \Frege proof.
		And by resolving disjuncts together we are left with $K^*$.
\item 
We can derive $(u \wedge\bigvee_{D\in \chi_0}^{\bar u\in D} \bar O_D )\vee K^*$ simply because $\bar O_K \vee O_K$ and $u \vee \bar u$ are tautological clauses and that $K^*$ which replaces every $S$-literal $l$ with $l^*$, contains $\bar u $ and $O_K$ as subclauses. The clause is derived using distributivity.

We replace each of these literals $l$ for $l^{u} \vee \bar u$. At the end of this process we can denote this clause as $K^{u}$ we do not have to include the extra $\bar u$ literal as that is already included.
$ u \vee \bar{\mathcal{O}}  \vee K^*\vee \bar K^{ u} $  is  a tautology which can be proven in a linear size \Frege proof. $K^*$ can be derived by resolving disjuncts.
\begin{prooftree}
	\AxiomC{$\bar O_K \vee O_K$}
	\UnaryInfC{${\mathcal{O}} \vee K^*$}
	\AxiomC{$u \vee \bar u$}
	\UnaryInfC{$u \vee K^*$}
	\AxiomC{$K^{u}$}
	\AxiomC{$\bar u \vee \bar{\mathcal{O}}\vee K^*\vee \bar K^{\bar u}$}
	\BinaryInfC{$\bar u \vee \bar{\mathcal{O}}\vee K^*$}
	\BinaryInfC{$\bar{\mathcal{O}}\vee K^*$}
	\BinaryInfC{$K^*$}
	\end{prooftree}

\item $K^*$ is a tautological clause.
\end{enumerate}
We are not yet done, because we have only derived clauses for $\chi_0$ not the whole of $\phi$. 
To do this we need a slightly better replacement scheme than $l$ to $l^*$.
To start with, for $\mathcal{S}$ variables $x$ where neither $x$ nor $\bar x$ appear in $L$, we can use trivial replacement  $x'=x$. Now consider $L= \pathl(\psi, C\vee u, \mathcal{S})$. For any literal $l$ such that $l, \bar l \in L$ then we use replacement $l$ to $l^*$. For any literal $l\in L$ such that $\bar l \notin L$, we use replacement $l$ to $l'$ where $l'= l^*\vee l$. Note that $l'$ and $l$ have the same dependency set.

Consider $\chi_0^*$ and weaken every $l^*$ literal such that $l\in L$ and $\bar l \notin L$ to $l \vee l^*$. 
Now consider all clauses in $\chi_1=\pathc(\psi,\phi\setminus\chi_0,\mathcal{S})$ and for each clause weaken every $l$ literal such that $l\in L$ and $\bar l \notin L$ to $l \vee l^*$. 
It remains to replace $\bar l \notin L$ when $l\in L$. We will have to use $\bar l'=\bar l^* \wedge \bar l$ to match  with $l'$. This can be done by distributivity in \Frege, if every clause $K\vee \bar l$  we have a copy of $K'\vee \bar l$ and $K' \vee \bar l^*$ where $K'$ is a copy of $K$ with every literal $l\in \mathcal{S}$ replaced by $l'$.

This turns out to be basically the case because $K$ can only contain  literals $k$ such that $k\in L$ and $\bar k\notin L$. $k\in L$ because we can extend the path from $C\vee u$ to $l$ through $\bar l$. If $\bar k \in L$ then $\bar l$ would be in $L$ as well, against our assumption. If $K\vee \bar l\in \chi_0\cap \chi_1$ we already have $K'\vee \bar l^*$ and $K' \vee \bar l$. It is possible that  $K\vee \bar l\in \chi_0$ but not $\chi_1$ in which case we create $K' \vee \bar l$ by weakening each $\mathcal{S}$ literal in $K\vee \bar l$ that is not $\bar l$.

Finally we need clause $C'$ which is the subclause $C$ where every $\mathcal{S}$-literal $l$ is replaced by $l'$. Note that all $\mathcal{S}$-literals in $C$ are by definition in $L$, so first we can derive $C^*$ and then weaken to get $C'$.
To do this we need the unit propagation proofs of $\phi \wedge \bar C \wedge \bar O_D \vdash_1 \bot $ for each $D\in \chi_0$ and $\bar l\in D$. We prove $\phi \wedge \bar C \wedge \bar O_D \rightarrow \bot $ in $O(|\phi|)$ many \Frege steps. We can then derive $ C \vee \bigwedge_{D\in \chi_0}^{\bar u\in D} O_D$ from $\phi$ next. Now we can replace all $\mathcal{S}$-literals $l$ in $C$ with $l^u\vee \bar u$ where $u \rightarrow(l^u\leftrightarrow l)$. This gets us $\bar u \vee \bigwedge_{D\in \chi_0}^{\bar u\in D} O_D \vee C^u$ i.e $\bar u \vee \bar {\mathcal{O}} \vee C^u$.
We also use $\bar u \rightarrow(l^{\bar u}\leftrightarrow l)$ and $C \vee u$ to replace every $\mathcal{S}$-literal $l$ with $l^{\bar u}$ and get $C^{\bar u} \vee u$. The $u$ literal can be reduced, (since this rule generalises the original universal reduction this is crucial).
$( u \wedge {\mathcal{O}}) \vee C^*\vee \bar C^{\bar u} $ is  a tautology once the definition of $l^*$ is considered which can be proven in a linear size Frege proof.
Likewise 	$ \bar u \vee \bar{\mathcal{O}} \vee C^*\vee \bar C^{ u} $  is  a tautology which can be proven in a linear size Frege proof.
And by resolving disjuncts together we are left with $C^*$. We can get always get $C'$ through weakening $C^*$.
\begin{minipage}{\textwidth}
\begin{prooftree}
	\AxiomC{$C^{\bar u}$}
	\AxiomC{$(u \wedge {\mathcal{O}})\vee C^*\vee \bar C^{\bar u}$}
	\BinaryInfC{$(u \wedge {\mathcal{O}})\vee C^*$}
	\AxiomC{$\bar u\vee \bar{\mathcal{O}}  \vee C^{ u}$}
	\AxiomC{$\bar u \vee \bar{\mathcal{O}}\vee C^*\vee \bar C^{u}$}
	\BinaryInfC{$\bar u \vee \bar{\mathcal{O}}\vee C^*$}
	\BinaryInfC{$C^*$}
	\UnaryInfC{$C'$}
\end{prooftree}
\end{minipage}\qed
\end{proof}
\newif\ifexample
\exampletrue
\ifexample
\begin{example}
	Consider the QBF with the prefix $\exists x \forall u \exists y, z, a$ and the matrix
	\[
		\psi = 
		(\overbrace{u \lor z                        }^{C \lor u}) \land
		(\overbrace{x \lor \bar u \lor y \lor \bar z}^{C_1}     ) \land
	    (\overbrace{\bar y \lor z \lor a     }^{C_2}     ) \land
		(\overbrace{x \lor \bar a                   }^{C_3}     ) \land
	    (\overbrace{x \lor y \lor z                 }^{C_4}     ) \land
	    (\overbrace{\bar y \lor \bar z              }^{C_5}     ) \land
        (\overbrace{\bar x \lor \bar u \lor a       }^{C_6}     )
	\]
	We have $\mathcal{S} = \{y, z, a\}$, $\chi_0=\pathc(\psi, C\vee u, \mathcal{S})=\{C \lor u, C_1, \dots, C_5\}$, and $L=\pathl(\psi, C\vee u, \mathcal{S})=\{y,\bar y, z, \bar z, a\}$.
	We are looking at outer clauses of those clauses in $\chi_0$ that contain $\bar u$, which is only $C_1$, with $O_{C_1} = \{x\}$.
	Notice that there is also the clause $C_6$ with $\bar u$ in it, but $C_6 \notin \chi_0$ because the only possible predecessor to $C_6$ is $C_3$ and it cannot be both entered and exited via its only $\mathcal{S}$-literal $\bar a$.
	We need to verify whether $C_1 \land \dots \land C_6 \land \bar C \land \bar O_{C_1}  \vdash_1 \bot$.
	Setting $\bar C \land \bar O_{C_1} = \bar x \land \bar z$ propagates $\bar a$ from $C_3$, then $\bar y$ from $C_2$, and finally $\bot$ from $C_4$.
	The preconditions of Lemma~\ref{lem:mixedred} are therefore satisfied, and $C$ can be soundly obtained by reduction from $C \lor u$.

	Let us now perform the substitutions from the proof of Lemma~\ref{lem:mixedred}.
	First, we introduce the extensions $y^{\bar u}, y^u, z^{\bar u}, z^u, a^{\bar u}, a^u$.
		 Next, we define the *-variables using $\mathcal{O} = \bigvee_{D \in \chi_0}^{\bar u \in D} \bar O_D = \bar x$:
	\begin{align*}
		\label{eq:star-defs}
		y^* = \left(\bar u \lor x \rightarrow y^{\bar u} \right) \land \left( u \land \bar x \rightarrow y^u \right) & &\quad
		z^* = \left(\bar u \lor x \rightarrow z^{\bar u} \right) \land \left( u \land \bar x \rightarrow z^u \right) \\
		a^* = \left(\bar u \lor x \rightarrow a^{\bar u} \right) \land \left( u \land \bar x \rightarrow a^u \right) \nonumber & &
	\end{align*}
	We have four types of clauses in $\chi_0$ depending on which literals of $u$ they contain.
	Let us start with the clauses that do not contain any $u$ literal, namely $C_2, C_3, C_4, C_5$.
	By resolution with the extension definitions and subsequent reduction of $u$, we obtain
	\begin{align*}
		C_2^u = & \; \bar y^u  \lor z^u \lor a^u & 
		C_3^u = & \; x \lor \bar a^u  &           
		C_4^u = & \; x \lor y^u \lor z^u   &
		C_5^u = & \; \bar y^u \lor \bar z^u  \\        
		C_2^{\bar u}  = & \;  \bar y^{\bar u} \lor z^{\bar u} \lor a^{\bar u} & C_3^{\bar u} = & \; x \lor \bar a^{\bar u} & C_4^{\bar u} = & \; x \lor y^{\bar u} \lor z^{\bar u} & C_5^{\bar u} = & \; \bar y^{\bar u} \lor \bar z^{\bar u}        
	\end{align*}
	For the sake of brevity, we shall demonstrate the next step only on $C_3$.
	Now $\bar u \lor \bar{\mathcal{O}} \lor C_3^* \lor \bar C_3^{\bar u} = \bar u \lor x \lor a^* \lor \bar a^u$
	can be simplified to $\bar u \lor x \lor a^u \lor \bar a^u$, which is obviously a tautology.
	It can be derived as follows: from the definition of $a^*$, we obtain $u \land \bar x \rightarrow a^* = a^u$, and thus $u \land \bar x \land \bar a^u \rightarrow \bar a^*$, or in other words, the clause above.
	Other *-replacements for the clauses $C_1, C_2, C_4, C_5$ are performed in a similar fashion.

	For the $y, z$ the variables $y^*, z^*$ are the final replacements, but for the variable $a$ we need to create the variable $a' = a^* \lor a$.
	We have $\chi_1 = \pathc\left( \psi, \{C_6\}, \mathcal{S} \right) = \{C_3, C_6\}$.
	We take $C_3^* = x \lor \bar a^*$ and $C_3 = x \lor \bar a$ and apply distributivity to obtain $C_3' = x \lor (\bar a^* \land \bar a)$.
	For other clauses that contain $a$ positively, we simply weaken to obtain $C_i'$. Note that if $C_3$ contained, say a $z$ literal, we would not be able to apply distributivity, however a literal like $z$ would also cause a path from $C$ to $C_6$. 

	Finally, we need to prove $C \lor \bar{\mathcal{O}} = x \lor z$, which we get by resolving the clauses on the unit propagation path: $C_4, C_3$, and $C_2$. 
	This will give us $x \lor z$.
	By resolving with the extension definitions for $z^u$ and $z^{\bar u}$, we obtain $u \lor x \lor z^{\bar u}$ and $\bar u \lor x \lor z^u$.
	Because $z^u$ and $z^{\bar u}$ do not depend on $u$, we can apply universal reduction to get the clauses $C^u = z^u$ and $C^{\bar u} = z^{\bar u}$.
	Finally, from the definition of $z^*$ we have $z^{\bar u} \rightarrow (z^* \leftrightarrow (u \land \bar x \rightarrow z^u))$, and thus $z^{\bar u} \land z^u \rightarrow z^*$, and we obtain the desired Frege proof of $C^* = \{ z^* \}$.
\end{example}

\fi

We have demonstrated ability to replace a conjunctive normal form with another and preserve DQBF satisfiability, by replacing variables with new variables from the Independent Extension rule.
Changing the formula, but preserving satisfiability is a key part of preprocessing, and this gives us a key connections to interference proof systems like DRAT, QRAT and DQRAT, whose introduction was to certify preprocessing steps.

\begin{lemma}\label{lem:spcases}
	Both QRATU and EUR are special cases of the rule p-simulated in Lemma~\ref{lem:mixedred}.
\end{lemma}

\begin{proof}
	\noindent\textbf{QRATU:} Suppose $C\vee u$ reduces to $C$ via QRATU. 
	We know for every clause $D\in \phi$ with $\bar u \in D$, $\phi\wedge \bar O_D \wedge \bar C \vdash_1 \bot$. Therefore no matter how many of these $D$ belong to $\chi_0$ in  Lemma~\ref{lem:mixedred}, every one satisfies the condition.
	
	\noindent\textbf{EUR:} Suppose $C\vee u$ reduces to $C$ via EUR. Via the side condition of EUR, $\chi_0$ does not contain any clauses $D\in \phi$ with $\bar u \in D$ therefore the side condition of the rule in  Lemma~\ref{lem:mixedred}.\qed
\end{proof}

With the lemmas we have shown, we now have the p-simulation of QRAT.  

\begin{corollary}
	\dFregeRed can p-simulate QRAT for false QBFs.
\end{corollary}

\begin{proof}
	When adding a clause through ATA. We simply use Lemma~\ref{lem:ATA} to add the clause.
	Instead of deleting a clause, we simply ignore it. \dFregeRed does not have any global conditions for any of its rules, except that a new extension rule require a new name for the variable, but this is immaterial to the proof complexity.
	For simulating QRATA, QRATU and EUR we create a new copy of the entire formula and then forget about clauses and variables outside of this. For QRATA we use Lemma~\ref{lem:QRATA} and for QRATU and EUR we use different special cases of Lemma~\ref{lem:mixedred} which is allowed by Lemma~\ref{lem:spcases}.\qed
\end{proof}

\subsection{P-simulation of the {\Drrs} Rule}\label{sec:qdrc}

The EUR rule uses resolution paths to calculate when it is safe to drop a single isolated literal. But in QBF and DQBF a more general approach can be taken, where resolution paths can be used to calculate that the dependency between a universal literal and a potential blocking existential literals is always spurious, in all clauses and throughout all stages of the proof.

The \Drrs relation \cite{Slivovsky-sat14} calculates $(u,x)$ only if there is a resolution path from $u$ to $\bar u$ through existential $x$. \Drrs is a sound way of modifying a DQBF prefix. We can add this to sound DQBF proof systems, in fact a number of QBF and DQBF proof system are already constructed in this way.

\subsubsection{Definition of \idrc}

\qdrc~\cite{Slivovsky-sat14} is a QBF proof system using the \Drrs rule, we can extend this to \idrc~\cite{BeyersdorffBlinkhornChewSchmidtSuda19} which uses \dirc as its proof system after modifying the proof system via \Drrs. The rules can be seen in Figure~\ref*{fig:idrc}.

\begin{figure}
	\framebox{\parbox{0.95\textwidth}{
			\begin{minipage}{\linewidth}
				\vspace{0.5em}
				{\centering\textbf{Axiom rule:} \axiom$(\phi)$\vspace{1em}\\}
				\begin{minipage}{.4\linewidth}
					\begin{prooftree}
						\AxiomC{}
						\UnaryInfC{$\{a^{\tau|_{D^{\textsf{rrs}}_a}} \mid a \in C, \var(a) \in \exists\}$}
					\end{prooftree}
				\end{minipage}
				\begin{minipage}{.55\linewidth}
					\begin{itemize}
						\item $C$ is a clause in the matrix $\phi$.
						\item $\tau|_{D^{\textsf{rrs}}_{a}} := \{\lnot l \mid l \in C, \var(l) \in \forall, (\var(l), \var(a)) \in D^{\textsf{rrs}} \}$.
					\end{itemize}
				\end{minipage}
			\end{minipage}
			\vspace{0.5em}\\
			\begin{minipage}{\linewidth}
				\vspace{1em}
				{\centering\textbf{Instantiation rule:} \instant$(C,\beta)$\\}
				\begin{minipage}{.4\linewidth}
					\begin{prooftree}
						\AxiomC{$C$}
						\UnaryInfC{$\{ a^{(\alpha \circ \beta)|_{D^{\textsf{rrs}}_a}} \mid a^\alpha \in C \}$}
					\end{prooftree}
				\end{minipage}
				\begin{minipage}{.55\linewidth}
					\begin{itemize}
						\item $\beta$ is a partial assignment to universal variables.
						\item$\alpha \circ \beta$ is $\{l \mid (l\in \alpha) \vee (l\in \beta \wedge l\notin\alpha)\}$
						\item$(\alpha \circ \beta)(a) := \{l \in \alpha \circ \beta \mid (\var(l), \var(a)) \in D^{\textsf{rrs}}\}$.
					\end{itemize}
				\end{minipage}
			\end{minipage}\\
			\begin{minipage}{\linewidth}
				\vspace{1em}
				{\centering\textbf{Resolution rule:} \res$(C_1,C_2,x^\tau)$\\}
				\begin{minipage}{.4\linewidth}
					\begin{prooftree}
						\AxiomC{$C_1$}
						\AxiomC{$C_2$}
						\BinaryInfC{$(C_1 \cup C_2) \setminus \{x^\tau,\lnot x^\tau\}$}.
					\end{prooftree}
				\end{minipage}
				\begin{minipage}{.55\linewidth}
					\begin{itemize}
						\item $x^\tau \in C_1$ and $\lnot x^\tau \in C_2$.
					\end{itemize}
				\end{minipage}
			\end{minipage}
	}}
	\caption{Proof rules of \idrc~\cite{BeyersdorffBlinkhornChewSchmidtSuda19} \label{fig:idrc}.}
\end{figure}

\subsubsection{Definition of DQRAT}

Another proof system that uses \Drrs is DQRAT. Here the \Drrs rule is a prefix modification rule used in conjunction with rules similar to QRAT. 
The rules of DQRAT are given in full in \cite{Blinkhorn2020SimulatingDP} and we provide a simplified overview here. 

\noindent\textbf{ATA:} Add a clause $C$ if $\phi\wedge \bar C \vdash_1 \bot$.\\
\textbf{DQRAT$_{\exists}$:} Add a clause $C$ if $l$ is an existential literal in $C$ and if for all $D\in \phi$ with  $\bar l \in D$ that $\phi\wedge \bar C\wedge O_D \vdash_1 \bot$. Where the outer clause $O_D$ with respect to $l$ is the set of all universal  literals $u$ in $D$, that are in $D_{\var(l)}$, the dependency set of $l$ and all existential literals $x$ in $D$, whose dependency set is no larger than $l$'s, i.e. $D_{\var(x)}\subset D_{\var(l)}$. \\
\textbf{UR:} Modify a clause $C\vee u$ to $C$ if $u$ is a universal literal and not in the dependency set of $C$.\\
\textbf{DQRAT$_{\forall}$:} Modify a clause $C\vee u$ to $C$, if $u$ is a universal literal and if for all $D\in \phi$ with  $\bar u \in D$ that $\phi\wedge \bar C\wedge O_D \vdash_1 \bot$. 
Where the outer clause $O_D$ with respect to $u$ is the set of all existential literals $x$ in $D$, that have variables in the set $S$ and the set of universal literals $v$ in $D$ that are in the dependency set of every variables in $S$ with $\var(v)\neq \var(u)$. $S$ is given as the set of all existential variables that depend on $u$.\\
\textbf{BPM:} A new variable is added to the prefix. In the existential case an arbitrary dependency set can be used.\\
\textbf{DRRS:} The dependency set of every existential variable can be modified to remove any spurious dependencies calculated by the Reflexive Resolution Path dependency scheme (\Drrs). \\
\textbf{DEL:} Delete a clause.
\subsubsection{P-simulations using \Drrs}
	
\begin{lemma}
	Suppose we have a DQBF $\forall U \exists E \psi$.
	We can calculate the relation \Drrs on $\forall U \exists E \psi$ between universal variables and existential variables. 
	Suppose there is a universal variable $u$ and we let  $\mathcal{S}$ define the set of these variables that depend on $u$,  $\mathcal{S}= \{x \in X \mid u \in D_x\}$.
	Using \dFregeRed we can derive the clauses of $\psi'$ that replaces each variable $x\in \mathcal{S}$ such that $(u,x)\notin \Drrs$ with $x'$ such that $D_{x'}\subset D_{x}\setminus{u}$.
\end{lemma}

\begin{proof}
	We focus on the variables $x\in \mathcal{S}$ that are declared independent to $u$ via \Drrs and we need to replace it and its negation in all clauses. If there is resolution path from $u$ to $x$ there is no resolution path from $\bar u$ to $\bar x$ and if there is a resolution path from $\bar u$ to $x$ there is no resolution path from $u$ to $\bar x$.
	We can write this formally, let $\chi_u$ be the subset of $\psi$ which contains every clause with a $u$ literal and let $\chi_{\bar u}$ be the subset of $\psi$ which contains every clause with a $\bar u$ literal. Define $L_{u}= \pathl(\psi, \chi_u, \mathcal{S})$
	and $L_{\bar u}= \pathl(\psi, \chi_{\bar u}, \mathcal{S})$, then for $(u,x)\notin \Drrs$, $(x\notin L_{u} \vee \bar x\notin L_{\bar u})\wedge(x\notin L_{\bar u} \vee \bar x\notin L_{ u})$ . This leaves the following possibilities (up to some symmetries):
	
	\begin{minipage}{0.95\textwidth}
		\begin{multicols}{2}
			\begin{enumerate}
				\item $x\notin L_{u}$, $\bar x\notin L_{u}$, $x\notin L_{\bar u}$ and $\bar x\notin L_{ \bar u}$
				\item $x\in L_{u}$, $\bar x\notin L_{u}$, $x\notin L_{\bar u}$ and $\bar x\notin L_{ \bar u}$
				\item $x\in L_{u}$, $\bar x\notin L_{u}$, $x\in L_{\bar u}$ and $\bar x\notin L_{ \bar u}$
				\item $x\in L_{u}$, $\bar x\in L_{u}$, $x\notin L_{\bar u}$ and $\bar x\notin L_{ \bar u}$
				
			\end{enumerate}
		\end{multicols}
	\end{minipage}
	\\
	
	We also need the symmetric cases for $\bar u$ in case 3 using  $u \rightarrow (x^{u}\leftrightarrow x)$ . 
	In cases 1, 2 and 3, we use both $\bar u \rightarrow (x^{\bar u}\leftrightarrow x)$ and $u \rightarrow (x^{u}\leftrightarrow x)$ to get clauses $u \vee \bar x \vee x^{\bar u}$ and $\bar u \vee \bar x \vee x^{ u}$. Now resolve them to get $\bar x\vee x^{\bar u} \vee x^{ u} $.
	Replace every $x$ literal with  $x^{\bar u} \vee x^{ u} $ via resolution. We consider $x'=x^{\bar u} \vee x^{ u}$.
	Now consider the presence of $\bar x$ in some clause $D\in \psi$.
	Every other $\mathcal{S}$-literal $l\in D$ cannot have $\bar l\in L_{u}$ nor $\bar l\in L_{\bar u}$, because it means $\bar x\in L_{ u}$ or
	$\bar x\in L_{\bar u}$ which cannot happen outside of case 4.
	We should have a copy of $D$ where every $\mathcal{S}$-literal $l$ is replaced by $l^u\vee l^{\bar u}$ except $\bar x$. We duplicate this and in one copy replace  $\bar x$ with $ u \vee \bar x^{\bar u}$, the $u$ literal can be reduced. In the other copy replace $\bar x$ with $ \bar u \vee \bar x^{ u}$ and here the $\bar u$ literal can be reduced. Now by taking a conjunction we effectively make the final $S$-literal replacement of this clause, which is to replace $\bar x$ with $\bar  x^{\bar u}\wedge \bar x^{ u}$.
	
	In Case 4 use definition $\bar u \rightarrow (x^{\bar u}\leftrightarrow x)$ to replace every $x$ literal with $u \vee x^{\bar u}$ and every $\bar x$ literal with $u \vee \bar x^{\bar u}$. 
	We have one remaining issue. Clauses $K$ with variables from case 4, may have an additional universal literal. 
	Firstly we claim that both $u$ and $\bar u$ literals are not present, this can only be introduced with another case 4 literal for $\bar u$ and this can only happen if there was a path from $\bar u$ to a literal $k$ in $K$ and a path from $\bar u$ to $\bar k$, but since $u \vee x^{\bar u}$, there would also be a path from $u$ to $k$ meaning $k$ is not case 4. 
	We also claim that you can always remove this universal literal $u$ with Lemma~\ref{lem:mixedred}, specifically the DQBF version of EUR. Suppose after the replacement there is a resolution path from $K$ to some clause with $\bar u$ in it, this $\bar u$ comes from either a clause originally with $\bar u $ in it, or a clause with a case 4 variable in it of the opposite polarity.
	Either way, there would be a resolution path from $\bar u$ to literal $x$ contradicting our assumptions in case 4.
	Therefore $x'=x^{\bar u}$ suffices.\qed
\end{proof}

\begin{corollary}
	\dFregeRed can p-simulate \idrc for false QBFs.
\end{corollary}

\begin{corollary}
	\dFregeRed can p-simulate DQRAT for false DBQFs.
\end{corollary}

\subsection{P-simulation of the Sequent Calculus}\label{sec:G}
Let $\Gamma$ and $\Delta$ each be sets of logical formulas. A \emph{sequent} $\Gamma \vdash \Delta$ expresses that any Boolean assignment to the free variables that satisfies every formula in $\Gamma$ also satisfies at least one formula in $\Delta$. A \emph{sequent calculus} is a proof system that uses sequents as lines. There are different sequent calculi, and sequent calculi can work for different logics. The QBF sequent calculus we are interested in is \Gfull, because of its position the the proof system hierarchy near the top (Figure~\ref{fig:QBFsimstruct}).

\subsubsection{\Gfull the QBF Sequent calculus} 

In a \Gfull sequent ($\Gamma \vdash \Delta$) $\Gamma$ and $\Delta$ are sets of QBFs. Note here that these QBFs are not necessarily in prenex form and they are also not necessarily closed, so they may contain a mix of bound and free variables.
The rules of \Gfull are given in Figure~\ref{fig_Gfull}. 
\begin{figure}[!h]
	\framebox{\parbox{\breite}
		{\small
			\begin{prooftree}
				\AxiomC{}
				\RightLabel{($\vdash$)}
				\UnaryInfC{$A\vdash A$}
				\DisplayProof\hspace{0.7cm}
				\AxiomC{}
				\RightLabel{($\bot\vdash$)}
				\UnaryInfC{$\bot\vdash$}
				\DisplayProof\hspace{0.7cm}
				\AxiomC{}
				\RightLabel{($\vdash\top$)}
				\UnaryInfC{$\vdash\top$}
			\end{prooftree}
			\begin{prooftree}
				\AxiomC{$\Gamma\vdash\Sigma$}
				\RightLabel{($\bullet\vdash$)}
				\UnaryInfC{$\Delta,\Gamma\vdash\Sigma$}
				\DisplayProof\hspace{1cm}
				\AxiomC{$\Gamma\vdash\Sigma$}
				\RightLabel{($\vdash\bullet$)}
				\UnaryInfC{$\Gamma\vdash\Sigma,\Delta$}
			\end{prooftree}
			\begin{prooftree}
				\AxiomC{$\Gamma\vdash\Sigma,A$}
				\RightLabel{($\neg\vdash$)}
				\UnaryInfC{$\neg A,\Gamma\vdash\Sigma$}
				\DisplayProof\hspace{0.7cm}
				\AxiomC{$A,\Gamma\vdash\Sigma$}
				\RightLabel{($\vdash\neg$)}
				\UnaryInfC{$\Gamma\vdash\Sigma,\neg A$}
			\end{prooftree}
			\begin{prooftree}
				\AxiomC{$A,\Gamma\vdash\Sigma$}
				\RightLabel{($\bullet\wedge\vdash$)}
				\UnaryInfC{$B\wedge A,\Gamma\vdash\Sigma$}
				\DisplayProof\hspace{0.7cm}
				\AxiomC{$A,\Gamma\vdash\Sigma$}
				\RightLabel{($\wedge\bullet\vdash$)}
				\UnaryInfC{$A\wedge B,\Gamma\vdash\Sigma$}
				\DisplayProof\hspace{0.7cm}
				\AxiomC{$\Gamma\vdash\Sigma,A$}
				\AxiomC{$\Lambda\vdash\Delta,B$}
				\RightLabel{($\vdash\wedge$)}
				\BinaryInfC{$\Gamma,\Lambda\vdash\Sigma_,\Delta ,A\wedge B$}
			\end{prooftree}
			\begin{prooftree}  
				
				
				\AxiomC{$\Gamma\vdash\Sigma,A$}
				\RightLabel{($\vdash\bullet\vee$)}
				\UnaryInfC{$\Gamma\vdash\Sigma,B\vee A$}
				\DisplayProof\hspace{0.5cm}
				\AxiomC{$\Gamma\vdash\Sigma,A$}
				\RightLabel{($\vdash\vee\bullet$)}
				\UnaryInfC{$\Gamma\vdash\Sigma,A\vee B$}
				\DisplayProof\hspace{0.7cm}
				\AxiomC{$A,\Gamma\vdash\Sigma$}
				\AxiomC{$B,\Lambda\vdash\Delta$}
				\RightLabel{($\vee\vdash$)}
				\BinaryInfC{$A\vee B,\Gamma, \Lambda\vdash\Sigma, \Delta$}
			\end{prooftree}
			\begin{prooftree}
				\AxiomC{$\Gamma\vdash\Sigma,A$}
				\AxiomC{$A,\Lambda\vdash\Delta$}
				\RightLabel{(cut)}
				\BinaryInfC{$\Gamma, \Lambda\vdash\Sigma, \Delta$}
			\end{prooftree}
			
			\begin{prooftree}
				\AxiomC{$A(B), \Gamma \vdash\Sigma$}
				\RightLabel{$(\forall \vdash)$}
				\UnaryInfC{$\forall x A(x),\Gamma   \vdash\Sigma$}
				\DisplayProof\hspace{0.7cm}
				\AxiomC{$\Gamma,\vdash\Sigma,  A(p)$}
				\RightLabel{$(\vdash\forall)$}
				\UnaryInfC{$\Gamma\vdash\Sigma, \forall x A(x)$}
			\end{prooftree}
			
			\begin{prooftree}
				\AxiomC{$A(p), \Gamma \vdash\Sigma$}
				\RightLabel{$(\exists \vdash)$}
				\UnaryInfC{$\exists x A(x),\Gamma   \vdash\Sigma$}
				\DisplayProof\hspace{0.7cm}
				\AxiomC{$\Gamma,\vdash\Sigma,  A(B)$}
				\RightLabel{$(\vdash\exists)$}
				\UnaryInfC{$\Gamma,   \vdash\Sigma, \exists x A(x)$}
			\end{prooftree}
			\centering{\text{ $A$, $B$ are QBFs and $\Gamma, \Lambda,\Sigma, \Delta$ are sets of QBFs,}  \text{Variable $p$ is not free on the lower sequents in ($\exists \vdash$), ($\vdash \forall$).} \\ \text{Variable $x$ replaces $p$ (or $B$) only and does not appear elsewhere in $A$.}
				\text{The free variables of $p$ and $B$ are not bound in $A$.}}
			
			\caption{Rules of the sequent calculus \Gfull \cite{KP90}.}
			\label{fig_Gfull} 
	}}
\end{figure}

\Gfull is a powerful QBF proof system that can p-simulate all the rules of QRAT and therefore does not have strategy extraction unless $\Pe=\PSPACE$. $\Gfull$'s strength comes from its ability to use QBFs as non-deterministic objects either as witnesses for quantifier introduction or from QBF cuts.

\subsubsection{QBFs as extension variables}
	Each sequent is made up of two sets of QBFs, a left hand side and a right hand side. As we are building up to a p-simulation by \dFregeRed, we will eventually represent a sequent as a clause, but first we will define each QBF as an extension variable using our independent extension variable rule:

\begin{definition}\label{def:extqbf}
We represent non-prenex QBF with potential free variables using universal variables and existential variables in the following recursive way:

\begin{itemize}
	\item the correct interpretation of atoms in sequents are universal variables (we can always add more universal variables to the prefix in \dFregeRed).
	\item $\neg \phi$. We take (possibly extension) variable $a$ representing $\phi$ and create new extension variable $c$ defined by $c \leftrightarrow (\neg a)$. $c$ represents $\neg \phi$.
	\item $\phi \wedge \psi$. We take (possibly extension) variables $a,b$ representing $\phi$ and $\psi$ respectively and create new  extension variable $c$ defined by $c \leftrightarrow (a \wedge b)$.  $c$ represents $\phi \wedge \psi$.
	\item $\phi \vee \psi$. We take (possibly extension) variables $a,b$ representing $\phi$ and $\psi$ respectively and create new  extension variable $c$ defined by $c \leftrightarrow (a \vee b)$. $c$ represents $\phi \vee \psi$.
	\item $\forall x \phi(x)$: 
	we take $a$ the existential variable representing $\phi$ and we make three definitions:
	$x\rightarrow (a^x\leftrightarrow a)$,\quad
	$\neg x\rightarrow (a^{\bar x}\leftrightarrow a)$,\quad $(a^{\forall x}\leftrightarrow a^{\bar x} \wedge a^x)$.
	$a^{\forall x}$ represents $\forall x \phi(x)$.
	\item $\exists x \phi(x)$: 
	we take $a$ the existential variable representing $\phi$ and we make three definitions:
	$x\rightarrow (a^x\leftrightarrow a)$,\quad
	$\neg x\rightarrow (a^{\bar x}\leftrightarrow a)$,\quad $(a^{\exists x}\leftrightarrow a^{\bar x} \vee a^x)$.
	$a^{\exists x}$ represents $\exists x \phi(x)$.

\end{itemize}

\begin{proposition}
	The dependency set $D_a$ of an extension variable $a$ representing a QBF $\phi$  in Definition~\ref{lem:aequiv} is exactly the set of free variables of $\phi$.
\end{proposition}

\begin{proof}
	For atom $u$, $u$ is its sole free variable and $D_u=\{u\}$.
	For propositional connectives $\neg, \wedge, \vee $ dependency and free variables are both given by the union of their defining subformulas.
	For quantification by $x$, WLOG $\forall x \phi(x)$, let $a$ be the extension variable representing $\phi(x)$ then $D_{a^{\forall x}}= D_a\setminus \{x\}$, but also the free variables of $\forall x \phi(x)$ are exactly to free variables of $\phi(x)$ minus $x$.\qed
\end{proof}

\end{definition}
\subsubsection{Sequents as clauses}

A sequent $\phi_1, \dots \phi_n \vdash \psi_1, \dots \psi_m$ can be represented by a clause:
$$\neg a_1 \vee\dots \vee \neg a_n \vee b_1 \vee\dots \vee b_m$$ where for $1\leq i \leq n$, $a_i$ is the variable representing $\phi_i$ and for $1\leq j \leq m$, $b_j$ is the variable representing $\psi_j$ using Definition~\ref{def:extqbf}.

\subsubsection{Sequent rules as clausal inferences}
With this interpretation, the propositional inferences fall into place very nicely, we can see this in Figure~\ref{fig_simG}, but this is a restatement of the well known fact that Frege rules can p-simulate the propositional fragment of \Gfull \cite{Kra95} (also known independently as LK).

\begin{lemma}[\cite{Kra95}]\label{lem:LK}
	All propositional rules ($\vdash$), ($\bot \vdash$),($\vdash\top$),($\bullet\vdash$), ($\vdash\bullet$) , ($\neg\vdash$), ($\vdash\neg$ ), ($\bullet\wedge\vdash$), ($\wedge\bullet\vdash$), ($\vdash\wedge$), ($\vdash\vee\bullet$), ($\vdash\bullet\vee$) ($\vee\vdash$) and (cut) can be p-simulated by \dFregeRed.
\end{lemma}

\begin{figure}[h]
	\framebox{\parbox{\breite}
		{\small
			\begin{prooftree}
				\AxiomC{}
				\RightLabel{($\vdash$)}
				\UnaryInfC{$\neg a\vee a$}
				\DisplayProof\hspace{0.7cm}
				\AxiomC{}
				\RightLabel{($\bot\vdash$)}
				\UnaryInfC{$\neg 0$}
				\DisplayProof\hspace{0.7cm}
				\AxiomC{}
				\RightLabel{($\vdash\top$)}
				\UnaryInfC{$1$}
			\end{prooftree}
			\begin{prooftree}
				\AxiomC{$\Gamma\vee\Sigma$}
				\RightLabel{($\bullet\vdash$)}
				\UnaryInfC{$ \Delta\vee \Gamma\vee\Sigma$}
				\DisplayProof\hspace{1cm}
				\AxiomC{$\Gamma\vee\Sigma$}
				\RightLabel{($\vdash\bullet$)}
				\UnaryInfC{$\Gamma\vee \Sigma\vee\Delta$}
			\end{prooftree}
			\begin{prooftree}
				\AxiomC{$ \Gamma\vee\Sigma\vee a$}
				\RightLabel{($\neg\vdash$)}
				\UnaryInfC{$\neg \neg a\vee \Gamma\vee\Sigma$}
				\DisplayProof\hspace{0.7cm}
				\AxiomC{$\neg a\vee\Gamma\vee\Sigma$}
				\RightLabel{($\vdash\neg$)}
				\UnaryInfC{$ \Gamma\vee\Sigma\vee\neg a$}
			\end{prooftree}
			\begin{prooftree}
				\AxiomC{$\neg a\vee \Gamma\vee\Sigma$}
				\RightLabel{($\bullet\wedge\vdash$)}
				\UnaryInfC{$\neg (b \wedge a)\vee 
					\Gamma\vee\Sigma$}
				\DisplayProof\hspace{0.7cm}
				\AxiomC{$\neg a\vee 
					 \Gamma\vee\Sigma$}
				\RightLabel{($\wedge\bullet\vdash$)}
				\UnaryInfC{$\neg (a\wedge b)\vee\Gamma\vee\Sigma$}
				\DisplayProof\hspace{0.7cm}
				\AxiomC{$ \Gamma\vee\Sigma\vee a$}
				\AxiomC{$ \Lambda\vee\Delta\vee b$}
				\RightLabel{($\vdash\wedge$)}
				\BinaryInfC{$ \Gamma\vee \Lambda\vee\Sigma\vee\Delta \vee(a\wedge b)$}
			\end{prooftree}
			\begin{prooftree}  
				
				
				\AxiomC{$ \Gamma\vee\Sigma\vee a$}
				\RightLabel{($\vdash\bullet\vee$)}
				\UnaryInfC{$ \Gamma\vee\Sigma\vee (b\vee a)$}
				\DisplayProof\hspace{0.5cm}
				\AxiomC{$\Gamma\vee\Sigma\vee a$}
				\RightLabel{($\vdash\vee\bullet$)}
				\UnaryInfC{$\Gamma\vee\Sigma\vee (a\vee b)$}
				\DisplayProof\hspace{0.7cm}
				\AxiomC{$\neg a\vee
					 \Gamma\vee\Sigma$}
				\AxiomC{$\neg b\vee 
					\Lambda\vee\Delta$}
				\RightLabel{($\vee\vdash$)}
				\BinaryInfC{$\neg (a\vee  b)\vee\Gamma\vee \Lambda\vee\Sigma\vee \Delta$}
			\end{prooftree}
			\begin{prooftree}
				\AxiomC{$\Gamma\vee\Sigma\vee a$}
				\AxiomC{$\neg a\vee \Lambda\vee\Delta$}
				\RightLabel{(cut)}
				\BinaryInfC{$\Gamma\vee \Lambda\vee\Sigma\vee \Delta$}
			\end{prooftree}
			
%
			\centering{\text{ $a$, $b$ are variables and $\Gamma, \Lambda,\Sigma, \Delta$ are subclauses} 
				}
			
			\caption{P-simulating rules of the sequent calculus LK }
			\label{fig_simG} 
	}}
\end{figure}

\subsubsection{Quantifier introduction rules}

Each quantifier introduction rule involves a substitution. For this we need to know that a certain substitutions $A[x/b]$ which replace $b$ with $x$, are equivalent to $A$ when $x=b$. This is not trivial (unlike the quantifier free case), because if $A$ contains quantifiers we can run into quantifier order problems. We use the following lemma with a suitable side condition to avoid this. 

\begin{lemma}\label{lem:aequiv}
	Let $x$ be a universal variable and let $b$ either be a universal variable or an existential variables with dependency set $D_b$. 
	Now we have two extension variables $a$ and $a'$ built on other extension variables: $\{t_1, \dots t_k\}$ and $\{t_1', ,t_k'\}$ where the definition of $t_i'$ only differs from $t_i$ by replacing inputs $t_j$ for $j<i$ with $t_j'$ and input $x$ with $b$, when these inputs occur.
	Suppose in every one of the extension clauses $\alpha\rightarrow (v \rightarrow B)$ that iteratively define $a$ and $t_i$ that neither variable $x$ nor any variable from $D_b$ appears in $\alpha$, then there is a $O(w^2 k)$ size \dFregeRed proof of 
	$(x \leftrightarrow b)\rightarrow(a \leftrightarrow a')$, where $w$ is the maximum width of the clauses involved. 
\end{lemma}

\begin{proof}
	We can demonstrate this by induction on the number of extensions needed to define $a$. 
	
	Let $Y$ be a set of literals possibly containing variables $x$ or $t_i$ for $1\leq i\leq k$ (the base case is when $k=0$). $Y'$ will be the same as $Y$ but replacing the variables $x$ with $b$ and $t_i$ with $t_i'$.
	Suppose, without loss of generality we use extension definition $\alpha\rightarrow (a \leftrightarrow \bigwedge_{y\in Y} y)$, then we also have extension definition $\alpha\rightarrow (a' \leftrightarrow \bigwedge_{y'\in Y'} y')$.
	Note that for each $y=t_i$ we can prove  
	$(x \leftrightarrow b)\rightarrow(t_i \leftrightarrow t_i')$ by the induction hypothesis, and so we prove  $(x \leftrightarrow b)\rightarrow(\bigwedge_{y\in Y} y \leftrightarrow \bigwedge_{y\in Y'} y')$ in $O(w)$ many steps.
	Using the extension definitions we get 
	$\alpha\wedge (x \leftrightarrow b)\rightarrow (a \leftrightarrow a')$.
	Our side condition is that $\alpha$ does not contain variables from $D_b$ nor $x$. 
	Additionally $\alpha$ does not contain any variables in $D_a$ nor $D_{a'}$ because the IndExt rule removes them, therefore we can reduce $\neg \alpha$ to get $(x \leftrightarrow b)\rightarrow (a \leftrightarrow a')$.\qed
\end{proof}

Now observe the right universal rule. Our derivation will look something like.

\begin{prooftree}
\AxiomC{$\Gamma \vee \Sigma \vee A$}
\UnaryInfC{$\Gamma \vee \Sigma \vee \forall x A[x/p]$}
\end{prooftree}
With universal variable $p$ not present in the dependency sets of $\Gamma$ and $\Sigma$, $x$ and $p$ not being bound in the definition of $A$. $A$ and $\forall x A[x/p]$ will be represented by an extension variable.

\begin{lemma}\label{lem:Aintro}
	\dFregeRed can p-simulate $(\vdash \forall)$ and $(\exists \vdash)$.
\end{lemma}

\begin{proof}

If we have derived a clause representing $\Gamma \vee \Sigma \vee A$, we have already represented the QBF $A$ with extension variable, which we refer here as $a'$, therefore we say our premise clause is $\Gamma \vee \Sigma \vee a'$.

We first define an extension variable $a$ which is constructed using the same definitions as $a'$ but replaces universal variable $p$ with universal variable $x$, i.e. $a= a'[x/p]$, unless of course a variable like this is already defined from some previous rule, then we would simply reuse this. 
In either case, we can prove by induction that $(p\leftrightarrow x)\rightarrow (a\leftrightarrow a')$  via Lemma~\ref{lem:aequiv}.

Next we define three independent extension variables (if they do not exist already):
$x\rightarrow (a^x\leftrightarrow a)$,\quad
$\neg x\rightarrow (a^{\bar x}\leftrightarrow a)$,\quad $(a^{\forall x}\leftrightarrow a^{\bar x} \wedge a^x)$.

Now we prove $\neg p \vee \neg x \vee \neg a' \vee a^x$
and 
$p \vee x \vee \neg a' \vee a^{\bar x}$. 
 Since we have independence for $x$ (but not $p$ yet) we can reduce to get 
 $\neg p \vee \neg a' \vee a^x$
 and 
 $p \vee \neg a' \vee a^{\bar x}$.
Now we resolve with the premise clause to get $\Gamma \vee \Sigma \vee \neg p  \vee a^x$ and $\Gamma \vee \Sigma \vee p \vee a^{\bar x}$. Since we now have independence for $p$ we reduce to get $\Gamma \vee \Sigma\vee a^x$ and $\Gamma \vee \Sigma \vee a^{\bar x}$. By distributivity we get $\Gamma \vee \Sigma\vee (a^x \wedge a^{\bar x})$. By definition we can replace $(a^x \wedge a^{\bar x})$ to get $\Gamma \vee \Sigma \vee a^{\forall x}$ where $a^{\forall x}$ is our extension variable representing $\forall x A[x/p]$. As with $\forall x A[x/p]$, $a^{\forall x}$ is independent of $x$ and $p$, but retains all free variables in its dependency set.

The proof for the left existential rule ($\exists\vdash$) is symmetric to the right universal rule ($\vdash\forall$).
\qed
\end{proof}

Now observe the right existential rule and how it would be represented as clauses.
\begin{prooftree}
	\AxiomC{$\Gamma \vee \Sigma \vee A$}
	\UnaryInfC{$\Gamma \vee \Sigma \vee \exists x A[x/b]$}
\end{prooftree}

With no variable in $D_b\cup \{x\}$ being quantified in  $A$.

\begin{lemma}\label{lem:Eintro}
	\dFregeRed can p-simulate $(\vdash \exists)$ and $(\forall \vdash)$.
\end{lemma}
\begin{proof}
We say our premise clause is $\Gamma \vee \Sigma \vee a'$, with $a'$ representing $A$.
 
Once again we define $a$ using the same circuit structure as $a'$ but replacing variable $b$ (which may be existential when an extension variable) with a fresh universal variable $x$. By Lemma~\ref{lem:aequiv}, we know $(x \leftrightarrow b)\rightarrow (a\leftrightarrow a')$ can be obtained.

Now we define:
$x\rightarrow (a^x\leftrightarrow a)$,\quad
$\neg x\rightarrow (a^{\bar x}\leftrightarrow a)$,\quad $(a^{\exists x}\leftrightarrow a^{\bar x} \vee a^x)$.
We next prove $\neg b \vee \neg x \vee \neg a' \vee a^x$.
and 
$b \vee x \vee \neg a' \vee a^{\bar x}$ and we reduce on $x$ to get $\neg b\vee \neg a' \vee a^x$
and 
$b\vee \neg a' \vee a^{\bar x}$, respectively. We can resolve the two over $b$ to get  $\neg a' \vee a^{\bar x} \vee a^{x}$ or even $\neg a' \vee a^{\exists x}$ once we use the definition. Now we can simply replace the $a'$ in $\Gamma \vee \Sigma \vee a'$ with $ a^{\exists x}$.

The left universal rule is symmetric to the right existential rule.\qed
\end{proof}

The combination of Lemmas~\ref{lem:LK}, \ref{lem:Aintro} and \ref{lem:Eintro} give us the following:

\begin{corollary}\label{cor:gsim}
	If there is a \Gfull proof $\pi$ of sequent $\mathcal{X}$, we can obtain its corresponding clause  with respect to Definition~\ref{def:extqbf} in a polynomial size number (in the size of $\pi$)  of \dFregeRed steps. 
\end{corollary}

\subsubsection{Refutational alignment}
Despite Corollary~\ref{cor:gsim} we have given ourselves the restriction that \dFregeRed is a refutational system that starts with axioms and a proof is only complete once we reach the symbol $0$.
Our interpretation of p-simulating $\Gfull$ is to say that for any $\Gfull$ proof of a sequent $\Pi \phi \vdash$, where $\Pi \phi$ is a closed QBF, with prefix $\Pi$ and matrix $\phi$ then we can construct in polynomial time an $\dFregeRed$ refutation of $\Pi \phi$. 

First it is useful to operate under a different set of variables in \Gfull, since we use universal variables for all atoms in Definition~\ref{def:extqbf}. We can unify them later using Lemma~\ref{lem:aequiv}.

\begin{lemma}\label{lem:newvar}
	Given a \Gfull proof of $\Pi \phi \vdash$, we can construct a $\Gfull$ proof of an isomorphic renaming of the variables of $\Pi \phi \vdash$ in polynomial times.
\end{lemma}
\begin{proof}
	We copy all proof steps under the renaming.\qed
\end{proof}

In \dFregeRed we only have access to formula $\phi$ to begin with, without quantifiers. Where our $\Gfull$ proof gives use $\neg \Pi \phi$ and our previous lemmas only give us a singleton clause that expresses this in a single extension literal, so we need some way to unravel the quantification to cause a contradiction.  

\begin{lemma}\label{lem:Aunfold}
	Suppose $\psi(Z,x)$ is a QBF with free variables $Z\sqcup\{x\}$ and also consider $\psi(Z,y)$ which replaces $x$ with $y$ so the free variables are $Z\sqcup\{y\}$.
	$x$ and $y$ are universally quantified in \dFregeRed.
	Let $s_x, s_y$ be extension variables representing $\psi(Z,x)$ and $\psi(Z,y)$, respectively using Definition~\ref{def:extqbf}. $s_x$ and $s_y$ must be defined isomorphically. 
	We use three additional definitions for quantification:
		$$\bar y \rightarrow (s_y^{\bar y}\leftrightarrow s_y), \quad y \rightarrow (s_y^y\leftrightarrow s_y), 
		\quad s_y^{\forall y}\leftrightarrow  s_y^{\bar y} \wedge s_y^y.
		$$
	Using these definitions there is a polynomial size \dFregeRed proof for 
$$(\neg s^{\forall y}_y \wedge (x \leftrightarrow \neg s^{ y}_y))\rightarrow \neg s_x$$
	 
\end{lemma}

\begin{proof}
	We prove $(x\leftrightarrow y)\rightarrow (s_x \leftrightarrow s_y)$ using Lemma~\ref{lem:aequiv}.
	Now it means we can derive $(\neg x \wedge \neg y) \rightarrow (s^y_y \leftrightarrow s_x)$ and 
	$( x \wedge y) \rightarrow (s^{\bar y}_y \leftrightarrow s_x)$.
	In both case we can reduce $y$ to get $\neg x \rightarrow (s^y_y \leftrightarrow s_x)$ and 
	$x \rightarrow (s^{\bar y}_y \leftrightarrow s_x)$.
	
	If we assume $\neg s_y^{\forall y}$ we get that $s^{\bar y}_y \rightarrow \neg s^y_y $. These are represented in the formulas we derive $\neg s_y^{\forall y}\wedge x \wedge \neg s_y^y \rightarrow \neg s_x$ and $\neg s_y^{\forall y}\wedge \neg x \wedge  s_y^y \rightarrow \neg s_x$. We can compose this together to get the formula 
	$(\neg s^{\forall y}_y \wedge (x \leftrightarrow \neg s^{ y}_y))\rightarrow \neg s_x$
	\qed
\end{proof}

The intuition behind $(x \leftrightarrow \neg s^{ y}_y)$ is that in order to remove the universal quantifier we need a Herbrand function. Being able to express a Herbrand function as a circuit would be strategy extraction, which is likely not possible for all $\Gfull$ proofs, instead we represent the strategy by a QBF witness which is encoded as an independent extension variable.

Note that when we repeat this idea for an existential variable $x$, we do not need the side condition of $x$ taking a strategy, as we will see.

\begin{lemma}\label{lem:Eunfold}
	Suppose $\psi(Z,x)$ is a QBF with free variables $Z\sqcup\{x\}$ and also consider $\psi(X,y)$ which replaces $x$ with $y$ so the free variables are $Z\sqcup\{y\}$.
	$x$ is existential, with $D_x$ a superset of the universal variables in $\bigcup_{z\in Z} D_z$.
	None of the variables of $D_x$ are bound in $\psi$.	
	$y$ is universally quantified and not contained in $\bigcup_{z\in Z} D_z$ nor in $D_x$. 
	Let $s_x, s_y$ be extension variables representing $\psi(X,x)$ and $\psi(X,y)$, respectively using Definition~\ref{def:extqbf}. $s_x$ and $s_y$ must be defined isomorphically. 
	We use three additional definitions for quantification:
	$$\bar y \rightarrow (s_y^{\bar y}\leftrightarrow s_y), \quad y \rightarrow (s_y^y\leftrightarrow s_y), 
	\quad s_y^{\exists y}\leftrightarrow  s_y^{\bar y} \vee s_y^y.
	$$
	Using these definitions there is a polynomial size \dFregeRed proof for 
	$$\neg s^{\exists y}_y\rightarrow \neg s_x.$$
	
\end{lemma}

\begin{proof}
	We start with Lemma~\ref{lem:aequiv} to show $(x\leftrightarrow y)\rightarrow(s_x \leftrightarrow s_y)$.
	This gives us $\neg x \vee \neg y \vee s^y_y \vee \neg s_x$ and  $ x \vee  y \vee s^y_{\bar y} \vee \neg s_x$ and in these we can reduce $y$ to get
	$\neg x \vee s^y_y \vee \neg s_x$ and  $ x \vee s^y_{\bar y} \vee \neg s_x$. Now we resolve over $x$ to get $s^y_{\bar y} \vee s^y_y \vee \neg s_x$ using the definition of $s_y^{\exists y}$ we get $\neg s^{\exists y}_y\rightarrow \neg s_x$.\qed
\end{proof}

\begin{theorem}
	\dFregeRed p-simulates \Gfull.
\end{theorem}

\begin{proof}
	Suppose we have a \Gfull proof of sequent $\Pi \phi \vdash$.
	We create a new set of universal variables to create and proof isomorphic sequent $\Pi' \phi' \vdash$.
	Using Definition~\ref{def:extqbf} to create an extension variable $d$ that represents $\Pi' \phi'$ and we can prove singleton clause $\neg d$.
	
	Note that in Lemma~\ref{lem:Aunfold} the variable $s^y_y$ does not contain $x$ in its dependency set. We consider it a strategy for $x$ and in the context we will refer to this  $\neg s^y_y$ as $\sigma_x$. So we will rewrite each use of Lemma~\ref{lem:Aunfold} as 	$(\neg s^{\forall y}_y \wedge (x \leftrightarrow \sigma_x))\rightarrow \neg s_x$.
	
	By using Lemma~\ref{lem:Aunfold} or Lemma~\ref{lem:Eunfold} we can remove the outer quantifier step by step. We eventually end up with 
	$\bigwedge_{x\in \forall, x\in \Pi} (x\leftrightarrow \sigma_x) \rightarrow \neg e$.
	Where $e$ is an extension variable that represents $\sigma$. Let the universal variables in $\Pi$ be $u_1 \dots u_n$ in order from left to right. Therefore we can simply derive disjunction $\bigvee_{i=1}^{n} (u_i\oplus \sigma_{u_i})$. The remaining proof proceeds similarly to the normal form originally used for \eFregeRed \cite{BBCP20}. So we realise that no existential variable in $\bigvee_{i=1}^{n} (u_i\oplus \sigma_{u_i})$ contains $u_n$ in its dependency set. Therefore we can reduce in both $0$ and $1$ to get $\sigma_{u_n}\vee \bigvee_{i=1}^{n-1} (u_i\oplus \sigma_{u_i})$ and $\neg \sigma_{u_n}\vee \bigvee_{i=1}^{n-1} (u_i\oplus \sigma_{u_i})$. We can resolve these two together to get $\bigvee_{i=1}^{n-1} (u_i\oplus \sigma_{u_i})$, so in this way we can repeatedly remove a disjunct, until reaching an empty disjunction i.e. $\bot$. 
	
	To remove any doubt we will also deal with QBFs that \Gfull can prove true i.e. there is a \Gfull proof of sequent $\vdash \Pi \phi$. We will say an \dFregeRed proof of $\vdash \Pi \phi$ is a refutation of $\bar \Pi \neg \phi$. 
	To do this we will show that a \Gfull proof of sequent $\vdash \Pi \phi$ can be extended easily to a \Gfull proof of sequent $\bar \Pi \neg \phi \vdash$.
	First we use \cite[Lem.~12]{CH22}  to show $\phi \vdash \phi$ and then the negation rule gives us $\phi, \neg \phi\vdash $. Now we add quantifiers to each QBF using the left rules. If we always begin with the left existential rule first we will satisfy the condition for the left universal rule, thus we proceed to introduce the quantifiers (this works the same way as the proof of \cite[Lem.~12]{CH22}) and we get $\Pi \phi, \bar \Pi \neg \phi \vdash$ the cut rule with $\vdash \Pi \phi$ gets us $\bar \Pi \neg \phi \vdash$ and we can then refute it in \dFregeRed by following the steps in this proof.\qed
\end{proof}

\begin{corollary}
	\dFregeRed p-simulates QRAT, on both true and false QBFs.
\end{corollary}

\section{Conclusion}\label{sec:conclusion}
We have introduced a powerful new proof rule that p-simulates most (D)QBF pre-processing and solving rules.
Some proof complexity questions remain, in particular whether \dFregeRed is reciprocally p-simulated by some of these proof systems.
We have not provided any direct p-simulations that work for long-distance Q resolution and its variations in QBF proof systems. 
A major obstacle is that long-distance resolution in DQBF is unsound. 
It will still be possible to represent these rules with side conditions that force the prefix to be ordered and like a QBF. 
Technically a p-simulation is already known because QRAT p-simulates long-distance Q-resolution via blocked literal eliminations (using QRATU) \cite{KHS17}. 
It may be possible that other simulations work via the methods we have used.

Other techniques that may have simulations include other dependency scheme rules such as the tautology free dependency scheme~\cite{BBP20}. One difficulty in simulating this currently is that not much is yet known about the Skolem transformations which have helped us find the right definitions for replacement variables.

\section*{Acknowledgments}
Thanks to Marijn Heule and Mikol\'{a}\v{s} Janota for their discussions. 
This work is supported by FWF Project ESP197.

\bibliographystyle{splncs04}
\bibliography{mainlncs}

\end{document}